\documentclass[10pt,a4paper,twocolumn]{article}

\usepackage[utf8]{inputenc}
\usepackage[english]{babel}
\usepackage{amsmath,amssymb}
\usepackage{booktabs}
\usepackage{siunitx}
\usepackage{graphicx}
\usepackage[hyphens]{url}
\usepackage{hyperref}
\usepackage[margin=1in]{geometry}
\usepackage{caption}
\usepackage{xcolor}
\usepackage{xcolor}
\usepackage{float}
\usepackage{array}
\usepackage{tabularx}
\usepackage{multirow}
\usepackage{xcolor}
\usepackage[activate=false]{microtype}
\usepackage{amsthm}
\usepackage{placeins}
\usepackage{mdframed}
\newmdenv[
  linecolor=green!40!black,
  backgroundcolor=green!4,
  innertopmargin=4pt, innerbottommargin=4pt,
  innerleftmargin=7pt, innerrightmargin=7pt,
  skipabove=5pt, skipbelow=5pt,
  linewidth=0.8pt
]{simbox}
\newtheorem{definition}{Definition}
\newtheorem{proposition}{Proposition}
\newtheorem{corollary}{Corollary}[proposition]
\newtheorem{lemma}{Lemma}
\newtheorem{theorem}{Theorem}
\newtheorem{remark}{Remark}
\usepackage{tikz}
\usetikzlibrary{arrows.meta, positioning, fit, calc, shadings}
\usepackage{pgfplots}
\pgfplotsset{compat=1.18}
\usepackage{pgfplotstable}

\newcommand{\dialsort}{\textsc{DialSort}}

\newcommand{\bigO}[1]{O\!\left(#1\right)}

\title{\textbf{DialSort: Non-Comparative Integer Sorting\\
via the Self-Indexing Principle\\[4pt]
\large Architecture, Implementation, and Substrate-Aware Analysis}}

\author{
Alexander Narvaez \\
Independent Researcher / Universidad EAFIT \\
Envigado, Antioquia, Colombia \\
\texttt{anarvae1@eafit.edu.co}
}

\date{April 2026}

\begin{document}

\maketitle

\begin{abstract}
Sorting over bounded-universe integer keys has traditionally relied on
counting sort and radix sort, both of which incur mandatory prefix-sum
passes, auxiliary scatter buffers, or multiple permutation passes.
This paper introduces \dialsort{}, a non-comparative sorting
architecture based on the \emph{self-indexing principle}: each integer
key~$k$ simultaneously encodes its value and its canonical position in
the ordered address space $[0,U-1]$.  \dialsort{} eliminates the
prefix-sum pass entirely by treating the histogram $H[\cdot]$ as the
canonical ordered representation, not as an intermediate structure.

To support parallel ingestion without serialization, we introduce the
\emph{Conflict Resolution Network} (CRN), a pipelined additive
reduction tree of depth $\lceil\log_2 k\rceil$ that resolves
concurrent writes using equality checks exclusively, with no magnitude
comparisons.  Formal proofs establish $O(n+U)$ sequential and
$O(n/k + \log k + U)$ parallel time bounds.

A software prototype on an 8-thread Intel x86-64 achieves
\textbf{39.77$\times$} speedup over \texttt{std::sort} and peak
throughput of \textbf{115.9~M~keys/s}.  Against Classic Counting Sort,
\dialsort{} wins 46 of 48 configurations (avg.\ 1.65$\times$ faster).
Against IPS$^4$o, \dialsort{} outperforms it in 24 of 48 sequential
and 29 of 48 parallel configurations, with speedups up to $9.1\times$
on uniform inputs.  Against ska\_sort, it wins 46 of 48 configurations
(avg.\ $3.33\times$, best $>7\times$).  All 208 benchmark
configurations passed correctness verification.

\dialsort{} is not a universal replacement for comparison-based
sorting, but a domain-specialised architecture for bounded-universe
workloads where sorting reduces to a geometric read over memory.
Benchmark source (v3.1) and five open interactive simulators are
released alongside this paper.
\end{abstract}

\section{Introduction}

Sorting is a fundamental operation in computer science, traditionally
studied under the comparison model, which imposes the well-known lower
bound of $\Omega(n \log n)$ for general inputs~\cite{knuth1998,cormen2009}.
Decades of research have produced highly efficient comparison-based
algorithms such as introsort and parallel sample sort~\cite{ips4o2022,sanders2011},
yet all of them rely on evaluating relative order between keys.

For bounded-universe integer keys, non-comparative methods such as
Counting Sort and Radix Sort~\cite{seward1954} bypass this lower bound
by exploiting structural properties of the input.  However, both still
incur significant overhead: Counting Sort requires a mandatory
prefix-sum pass and an auxiliary output array, while Radix Sort
performs multiple permutation passes.  These extra steps introduce
memory traffic, synchronization points, and sequential dependencies
that limit scalability, especially in parallel settings.

This paper introduces \dialsort{}, a non-comparative sorting
architecture grounded in the \emph{self-indexing principle}: an integer
key inherently encodes its own position in the ordered address space.
Rather than a new sorting paradigm, \dialsort{} constitutes a
\textbf{high-performance architectural contribution}: it eliminates the
mandatory prefix-sum pass, introduces the CRN as a parallel
conflict-resolution primitive with bounded latency, and enables a
substrate-aware analysis spanning CPU software, FPGA, photonics, and
wafer-scale processors.  The self-indexing observation is related to
\emph{direct-address tables}~\cite[Ch.\,11]{cormen2009}; its
systematic architectural exploitation is the novel contribution.

To support efficient parallelism, we introduce the \emph{Conflict
Resolution Network} (CRN), a pipelined additive reduction tree that
resolves concurrent writes using only equality checks.  Unlike
prefix-sum-based approaches, the CRN operates during ingestion,
eliminating global synchronization and delivering bounded latency even
under highly skewed distributions.

We distinguish \dialsort{} from classical Counting Sort both
conceptually and empirically.  Conceptually, \dialsort{} redefines the
histogram from an intermediate counting structure into the canonical
ordered state itself, making output reconstruction optional.
Empirically, a direct side-by-side comparison across 48 configurations
shows that \dialsort{} outperforms classical Counting Sort in
46 cases, with an average speedup of 1.65$\times$ (up to $4.3\times$
in the best case).

Experimental results on commodity hardware demonstrate that \dialsort{}
achieves up to 39.77$\times$ speedup over \texttt{std::sort} for
bounded-universe inputs and 14.61$\times$ using a radix-based extension
for the full int32 range.  A direct comparison against IPS$^4$o
(sequential), one of the fastest comparison-based sorters in practice,
shows \dialsort{} outperforming IPS$^4$o-seq in \textbf{24 of 48}
sequential configurations and \dialsort{}-Parallel outperforming
IPS$^4$o-par in \textbf{29 of 48} parallel configurations, with
speedups reaching $9.1\times$ (sequential, $U{=}65536$, uniform) and
$4.08\times$ (parallel, reverse-sorted) at $N{=}10^7$.  Performance is
strongly regime-dependent: \dialsort{} excels on uniform and skewed
distributions, while IPS$^4$o retains clear superiority on sorted and
reverse-sorted inputs by exploiting pre-existing order.  A direct
comparison against ska\_sort, one of the fastest non-comparative CPU
sorters in practice, shows \dialsort{} outperforming it in 46 of 48
configurations, with an average speedup of $3.33\times$.  These results
indicate that sorting in this domain can be interpreted as a geometric
process of revealing order already embedded in the data representation,
with \dialsort{} functioning as a competitive domain-specialized
alternative rather than a universal replacement.

\medskip
\noindent\textbf{Contributions:}
\begin{enumerate}
\item A formalization of the self-indexing principle for any finite
      totally ordered domain (Proposition~2), with rigorous
      order-theoretic foundations for the canonical embedding (Proposition~1) and
      order-embedding vs.\ order-isomorphism.
\item The Conflict Resolution Network (CRN) for parallel,
      conflict-free accumulation without prefix sums, with fully
      defined quantities, explicit topology, and formal proofs.
\item Formal complexity proofs for sequential ($O(n+U)$) and
      parallel ($O(n/k + \log k + U)$) variants.
\item Both structural and empirical differentiation from classical
      Counting Sort.
\item A regime-stratified empirical comparison against IPS$^4$o:
      \dialsort{} outperforms IPS$^4$o-seq in 24 of 48 sequential
      configurations and \dialsort{}-Parallel outperforms IPS$^4$o-par
      in 29 of 48 parallel configurations; \dialsort{} dominates on
      uniform/skewed inputs, IPS$^4$o dominates on ordered inputs.
\item A direct empirical comparison against ska\_sort, the
      state-of-the-art non-comparative CPU sorter.
\item A fair parallel matchup: \dialsort{}-Parallel vs.\
      IPS$^4$o-par (both 8 threads, 48 configurations),
      showing \dialsort{}-Parallel winning in 29 of 48 cases
      with average ratio $1.90\times$ (Section~\ref{sec:vs_ips4o_par}).
\item An analytical projection on Cerebras WSE-3 (Substrate-4,
      Proposition~3) showing that ingestion latency is bounded by a
      hardware constant of the fixed mesh topology, with the ordered
      scan $O(U)$ and host-device transfer as dominant practical terms.
\item Five open interactive simulators covering four substrates and the CRN pipeline (S1\,\&\,S2 share one simulator; S3, S4, and the CRN each have a dedicated simulator).
\end{enumerate}

\begin{quote}
\textbf{\dialsort{} does not compute order.  It reveals it.}
\end{quote}

\noindent Formally: the algorithm computes multiplicities via direct
indexing and then reads the pre-existing total order of the address
space through a monotone scan over the address space $[0,U-1]$,
with zero-count cells contributing no output, without evaluating any
order comparison between input keys.

\section{Related Work and Formal Differentiation\label{sec:related}}

\subsection{Existing non-comparative sorts}

Counting sort and radix sort \cite{seward1954} exploit non-comparative
properties of integer keys but require additional steps: counting sort
mandates a prefix-sum pass and mandatory output reconstruction; LSD
radix sort \cite{cormen2009} introduces multiple permutation passes.
Sorting networks \cite{batcher1968} provide fixed-structure ordering but
still rely on comparison primitives at every node.  Bead sort
\cite{beadsort2012} offers a physical analogy but lacks a computational
architecture and requires $O(n^2)$ space.  Engineering-oriented
implementations such as ska\_sort and IPS$^4$o \cite{ips4o2022} are the
practical state-of-the-art for CPU sorting; IPS$^4$o uses order
comparisons as fallbacks, while ska\_sort is itself non-comparative but
requires digit decomposition and scatter passes.  IPS$^4$o is
benchmarked directly in Section~\ref{sec:vs_ips4o}; ska\_sort is
benchmarked directly in Section~\ref{sec:vs_ska}.

Pigeonhole sort~\cite{sedgewick2011} achieves the same $O(n+U)$ complexity
but stores element lists per bucket rather than counts, making it
structurally equivalent to a non-stable variant of counting sort; our
differentiation from counting sort (Section~\ref{sec:related}) therefore
subsumes it.

\subsection{Direct addressing and self-indexing}

The observation that an integer key can serve as its own array index is
well established as \emph{direct-address tables}~\cite[Ch.\,11]{cormen2009}.
\dialsort{} situates itself explicitly in this context: the
mathematical duality between value and index is not claimed as a new
discovery.  The novel contribution is the systematic architectural
exploitation of this duality---CRN, elimination of prefix-sum, and the
substrate-aware analysis across four physical substrates.

\subsection{RAM-model complexity bounds for integer sorting}

In the word-RAM model with word size $w$ bits, sorting $n$ integers
from $[0,U{-}1]$ can be accomplished in $O(n \log \log U)$
time~\cite{han2002}.  Counting sort achieves $O(n{+}U)$ time and
$O(U)$ space, optimal when $U = O(n)$.  \dialsort{} operates in the
same asymptotic complexity class as counting sort for the sequential
case and is therefore optimal in the same regime ($U = O(n)$).  For
$U \gg n$, word-RAM methods such as Han--Thorup~\cite{han2002} achieve
superior asymptotic bounds; \dialsort{} does not compete in that regime
and explicitly targets $n \gtrsim 10 \cdot U$ (Section~\ref{sec:limits}).
The contribution of \dialsort{} lies not in beating Han--Thorup bounds,
but in the parallel CRN architecture achieving $O(n/k + \log k + U)$
wall-clock time and in the substrate-aware analysis across four physical
execution paradigms.

\subsection{Prior work on self-indexed sorting}

A related idea appears in the 1996 Self-Indexed Sort (SIS)
proposal~\cite{wang1996}, which is essentially a linked-list variant
of classical counting sort.  Unlike that work, \dialsort{} redefines
the histogram as the canonical ordered state, introduces the Conflict
Resolution Network for native parallelism, and provides both formal
(Proposition~1) and empirical differentiation (46 of 48 configurations
faster, average 1.65$\times$).

\subsection{Structural and architectural differentiation from counting sort}

The most predictable reviewer objection is that \dialsort{} is ``just
counting sort renamed.''  Table~\ref{tab:related} addresses this
formally.  We state the distinction explicitly:

\begin{quote}
\dialsort{} is not a variant of counting sort: it is a different
execution model over the same mathematical object.
Counting sort uses the histogram as a temporary intermediate structure
to be consumed by a mandatory prefix-sum pass and output-scatter phase.
\dialsort{} treats the histogram as the canonical ordered
representation itself; output reconstruction is optional and
interface-dependent, not algorithmic.  This distinction eliminates one
full pass over $O(U)$ memory and the $O(n)$ scatter buffer.
\end{quote}

\begin{table*}[tp]
\centering
\caption{Structural comparison against related non-comparative sorts.
$\star$~=~\dialsort{}-specific contribution.}
\label{tab:related}
\small
\begin{tabular}{@{}lp{3.8cm}p{3.2cm}p{4.2cm}@{}}
\toprule
\textbf{Property} & \textbf{Classic Counting Sort} & \textbf{Radix Sort (LSD)} & \textbf{\dialsort{} (this work)} \\
\midrule
Core principle      & Freq.\ counting + prefix sum  & Digit decomposition, $k$ passes & Self-indexing duality $\star$ \\
Prefix sum required & Yes (mandatory)               & Per digit                        & \textbf{No} $\star$ \\
Output vector       & Always required               & Always required                  & \textbf{Optional} $\star$ \\
Passes over data    & 3 (count + prefix-sum + scatter) & $4k$ (base-256)              & \textbf{2} (ingest + project) $\star$ \\
Auxiliary memory    & $O(U+n)$                      & $O(n)$ extra                     & $O(U)$ sequential; $O(U+k)$ with CRN $\star$ \\
Parallel model      & Locks / atomics (weak)        & Per-digit parallelism            & Native CRN $\star$ \\
Conflict handling   & None / serialized             & None                             & Additive CRN $\star$ \\
Large-$U$ support   & Not addressed                 & Handles large $U$ via digits     & Hierarchical (future work) \\
Formal model        & Algorithmic                   & Algorithmic                      & Geometric + physical-state $\star$ \\
\bottomrule
\end{tabular}
\end{table*}

\subsection{Order comparison vs.\ presence test}

A reviewer may argue that checking $H[k] \neq 0$ is a comparison.
Table~\ref{tab:comparison} formalizes why this is not an order
comparison in the relevant sense.

\begin{table*}[tp]
\centering
\caption{Formal distinction between order comparisons (which \dialsort{} eliminates) and
presence tests (which \dialsort{} uses).  The distinction has direct physical consequences
in FPGA and photonic substrates.}
\label{tab:comparison}
\small
\begin{tabular}{@{}lp{5.5cm}p{5.5cm}@{}}
\toprule
\textbf{Property} &
\textbf{Order comparison ($x_i > x_j$)} &
\textbf{Occupancy query on histogram state ($H[k] \neq 0$)} \\
\midrule
Operands   & Two distinct input keys $x_i$, $x_j$      & One histogram cell $H[k]$ \\
Determines & Relative order between keys                & Occupancy of a single memory cell \\
Complexity & $\Omega(n\log n)$ unavoidable in comparison model & $O(1)$ per cell; no key vs.\ key involved \\
FPGA       & Requires a comparator circuit              & Register state \emph{is} the answer --- no test evaluated \\
Photonic   & Requires interaction between key signals   & Active resonator conducts; inactive is transparent \\
\midrule
\multicolumn{2}{l}{\dialsort{} \textbf{eliminates} this operation $\longrightarrow$} &
\dialsort{} \textbf{uses} this (not an order comparison under
Definition~\ref{def:ordercomp}) $\longrightarrow$ \\
\bottomrule
\end{tabular}
\end{table*}

\section{The Self-Indexing Principle}

\subsection{Order-theoretic foundations\label{sec:foundations}}

We first establish the order-theoretic setting precisely, responding to
the reviewer's observations regarding order relations.

\begin{definition}[Total order on the universe]
Let $\mathcal{U} = \{0,1,\ldots,M{-}1\} \subset \mathbb{Z}_{\geq 0}$
for some positive integer $M$.  The relation $\leq$ defined by
$a \leq b \iff b - a \in \mathbb{Z}_{\geq 0}$ is a \emph{total order}
on $\mathcal{U}$: it is reflexive ($a \leq a$), antisymmetric
($a \leq b$ and $b \leq a \Rightarrow a = b$), transitive, and total
(every pair is comparable).
\end{definition}

\begin{remark}
Throughout this paper, ``order'' refers to the total order $(\leq)$,
not to the strict relation $(<)$.  The strict order $(<)$ is a
\emph{strict total order} (irreflexive, asymmetric, transitive) but
not a partial order in the standard sense, since it fails reflexivity.
Algorithms exploit a pre-existing order structure; they do not define
it.
\end{remark}

\begin{definition}[Self-Indexing Duality]
For any integer $k \in [0, U-1]$, $k$ is simultaneously a value and
its own index in the naturally totally ordered address space
$([0,U-1],\leq)$.  No comparison-based computation is required to
determine the value-address correspondence; the total order $\leq$ is
a mathematical object independent of any algorithm, and the remaining
work is the ordered projection over the address space.
\end{definition}

\subsection{Input Multiset, Support Set, and Embedding Conditions\label{sec:iso}}

Let $K$ be the input multiset of keys drawn from the bounded universe
$[0,M-1]$.  Define its \emph{multiplicity function}
\[
H : [0,M-1] \to \mathbb{N}, \qquad
H(k) = |\{x \in K : x = k\}|.
\]
Define the \emph{support set}
\[
D = \{k \in [0,M-1] \mid H(k) > 0\}.
\]
The distinction between $K$, $H$, and $D$ is fundamental. The multiset
$K$ carries multiplicity information; the function $H$ records that
multiplicity; and $D$ is the order-theoretic object relevant for
embedding values into the ordered address space.  Accordingly,
order-theoretic statements are made over $D$, while repeated
occurrences are represented exclusively through $H$.

\begin{definition}[Order comparison\label{def:ordercomp}]
An \emph{order comparison} is any evaluation of a relation in
$\{<,\leq,>,\geq\}$ between two input keys, or between two values
derived from input keys for the purpose of establishing their relative
order.  Equality tests $x = y$, direct indexing $H[x]$, and occupancy
queries on histogram cells are \emph{not} order comparisons under this
definition, because they do not determine the relative order of two
distinct keys.
\end{definition}

\begin{definition}[Canonical embedding\label{def:embedding}]
Define
\[
\phi : D \to [0,M-1], \qquad \phi(k) = k.
\]
Since $D \subseteq [0,M-1]$, this map is well defined.
\end{definition}

\begin{proposition}[Embedding and isomorphism conditions\label{prop:iso}]
The map $\phi$ is an \emph{order-embedding} from $(D,\leq)$ into
$([0,M-1],\leq)$.  It is an order-isomorphism if and only if
\[
D = [0,M-1].
\]
\end{proposition}

\begin{proof}
Because $\phi(k)=k$, if $a,b \in D$ and $a \leq b$, then
$\phi(a) = a \leq b = \phi(b)$, so $\phi$ is order-preserving.  It is
also injective, again because $\phi(k)=k$.  Hence $\phi$ is an
order-embedding.  It is an order-isomorphism exactly when it is
bijective onto $[0,M-1]$, which holds if and only if every element of
the universe appears in the input, i.e., $D = [0,M-1]$.
\end{proof}

\begin{corollary}\label{cor:general}
In the general case, \dialsort{} relies on the order-embedding
$\phi : D \hookrightarrow [0,M-1]$ (Proposition~\ref{prop:iso}),
with multiplicities stored separately in $H$.  The sorted order is
therefore read from the pre-existing total order of the address space
together with the multiplicity profile encoded in $H$, rather than
constructed by pairwise key comparisons.
\end{corollary}

\begin{corollary}[Isomorphism in Substrate-3\label{cor:photonic}]
$\phi : D \to [0,M-1]$ is an order-isomorphism whenever $D = [0,M-1]$
(Proposition~\ref{prop:iso}), regardless of multiplicities, since
duplicates live in $H$, not in $D$.  Substrate-3 additionally achieves
a \emph{physical} one-to-one correspondence when $H[k] \in \{0,1\}$
for all $k$: each active resonator corresponds bijectively to one
element of $D$ and preserves $\leq$.  When $H[k] \geq 2$ for some
$k$, the mathematical isomorphism on $D$ is unaffected but the physical
bijection breaks.  Substrate-3 is the unique substrate where both hold
simultaneously.
\end{corollary}

\begin{remark}
In Substrate-4 (WSE-3, Section~\ref{sec:s4}), each tile maintains
$H[k] \geq 0$ in its local SRAM.  The mathematical order-isomorphism
$\phi : D \to [0,M-1]$ holds whenever $D = [0,M-1]$; when
$D \subsetneq [0,M-1]$, $\phi$ is an order-embedding.  The physical
one-to-one resonator correspondence of Substrate-3 is absent in all
other substrates.
\end{remark}

\subsection{Mathematical Formulation, Correctness, and Ordered Projection}

Let $K$ be a multiset of $n$ keys in $[0,U-1]$.  Construct a histogram
\[
H : [0,U-1] \to \mathbb{N}
\]
initialised to zero, and perform the ingestion step
\begin{equation}
\forall x \in K : \quad H[x] \leftarrow H[x] + 1.
\label{eq:ingestion}
\end{equation}

No magnitude comparison between pairs of input keys is required in this
step (Definition~\ref{def:ordercomp}): each key is used directly as an
address.

After ingestion, define the ordered projection by scanning the universe
in ascending order and emitting each key according to its multiplicity:
\begin{equation}
\text{for } k = 0,1,\ldots,U-1, \quad \text{emit } k
\text{ exactly } H[k] \text{ times.}
\label{eq:projection}
\end{equation}

The support set of distinct values is $D = \{k \in [0,U-1] \mid H(k)>0\}$,
and the canonical order-embedding is $\phi : D \to [0,U-1]$, $\phi(k)=k$
(Proposition~\ref{prop:iso}).

\begin{theorem}[Correctness of \dialsort{}\label{thm:correct}]
Let $K$ be a multiset of keys in $[0,U-1]$, and let $H$ be the
histogram obtained by~\eqref{eq:ingestion}.  The sequence produced
by~\eqref{eq:projection} is:
\begin{enumerate}
  \item nondecreasing;
  \item equal to $K$ as a multiset;
  \item of total length $n = |K|$.
\end{enumerate}
\end{theorem}

\begin{proof}
By construction, $H[k]$ equals the multiplicity of key $k$ in $K$.
The projection scans values in increasing order of $k$; therefore the
emitted sequence is nondecreasing.  Since each key $k$ is emitted
exactly $H[k]$ times, the output multiset coincides with $K$.
Finally,
\[
\sum_{k=0}^{U-1} H[k] = n,
\]
so the total number of emitted elements is exactly $n$.
\end{proof}

\textbf{Canonical ordered state.}
The histogram $H[\cdot]$ is the canonical ordered state of the dataset:
it stores multiplicity over an already ordered address space.  The
linear output vector is therefore a projection of that state, not the
primary ordered representation itself.  Direct consumers of $H$ can
answer $\mathrm{frequency}(k)$ in $O(1)$, $\mathrm{presence}(k)$ in
$O(1)$, and $\mathrm{range\mbox{-}count}(a,b)$ in $O(b-a)$ without
reconstructing a linear output array.

\subsection{General Self-Indexing Principle}

The self-indexing duality is not exclusive to machine integers.
It holds for any finite totally ordered domain that admits an
order-preserving mapping to a bounded index space.

\begin{proposition}[General Self-Indexing Principle]
Let $(S, <_S)$ be a finite totally ordered set, and let
$\phi: S \to \{0, 1, \dots, U-1\}$ be a known \emph{order-isomorphism}
(i.e., a bijective order-homomorphism) such that $\phi$ and $\phi^{-1}$
are computable in times $T_\phi$ and $T_{\phi^{-1}}$ respectively:
\[
s_1 <_S s_2 \;\iff\; \phi(s_1) < \phi(s_2).
\]
Note that $\phi$ is an order-isomorphism here \emph{by hypothesis of
the proposition}, not a property claimed for the general canonical
embedding $\phi : D \hookrightarrow [0,M-1]$ of
Definition~\ref{def:embedding}, which is an order-embedding in the
general case.  Then any multiset $K$ of elements of $S$ can be sorted
without performing any order comparison
(Definition~\ref{def:ordercomp}) between its elements.  The sorting
cost is $O(n T_\phi + U T_{\phi^{-1}} + n + U)$; when
$T_\phi = T_{\phi^{-1}} = O(1)$ (the direct-address case), this
reduces to $O(n+U)$.
\end{proposition}

\begin{proof}[Constructive proof (\dialsort{})]
\begin{enumerate}
\item For each $s \in K$, compute $i = \phi(s)$ and increment
      $H[i]$\texttt{++}.
      No order comparison is evaluated.
\item Sweep $i = 0, \dots, U-1$ and emit $\phi^{-1}(i)$ exactly
      $H[i]$ times.
      No order comparison is evaluated.
      When $S$ is the integers, $\phi^{-1}(i) = i$ directly;
      for other domains a lookup table suffices.
\item The emitted sequence is ordered by $<_S$ because $\phi$
      preserves order.
\end{enumerate}
\end{proof}

\begin{corollary}
The order of $K$ is not constructed by the algorithm.
It exists implicitly in the index space defined by $\phi$, and the
geometric scan reads it directly without evaluating any magnitude
comparison.
\end{corollary}

This proposition extends \dialsort{} beyond integers.
Any finite totally ordered domain that admits a known order-isomorphism
to a bounded integer range can be sorted without comparisons using the
same geometric mechanism.  In all such cases $\phi$ is a bijection
satisfying the isomorphism conditions of Proposition~\ref{prop:iso};
when $\phi$ is only an order-embedding (the general integer case,
where $D \subsetneq [0,U-1]$), the algorithm still operates correctly
but $\phi$ is not an isomorphism (Corollary~\ref{cor:general}).
Examples of domains with a natural order-isomorphism include
ASCII and Unicode characters ($\phi$ = code point),
calendar days ($\phi$ = day number, $U = 366$),
MIDI musical notes ($\phi$ = note number, $U = 128$),
grayscale pixel values ($\phi$ = intensity, $U = 256$),
and any finite enumeration with a natural linear order.

Figure~\ref{fig:dialsort_crn_pipeline} summarizes the full \dialsort{}
dataflow: direct ingestion, conflict-free additive reduction through the
CRN, canonical storage in $H[k]$, and optional linear output
projection.

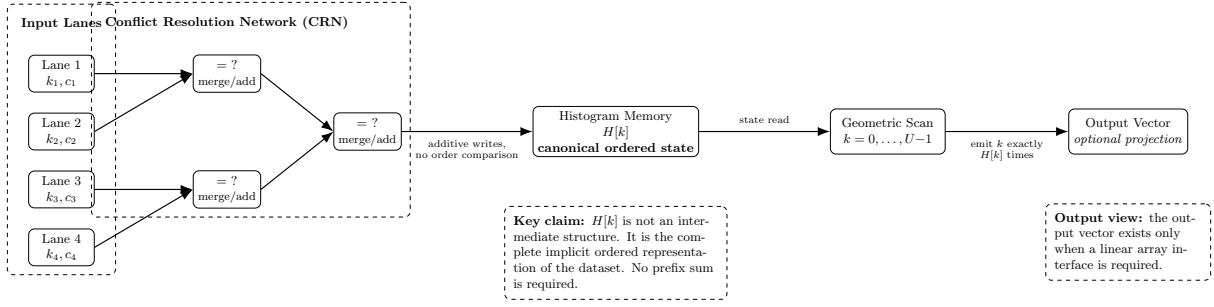
\begin{figure*}[tp]
\centering
\resizebox{\textwidth}{!}{%
\begin{tikzpicture}[
    font=\small,
    block/.style={
        draw, rounded corners,
        minimum width=2.8cm, minimum height=1.1cm, align=center
    },
    smallblock/.style={
        draw, rounded corners,
        minimum width=1.6cm, minimum height=0.75cm, align=center
    },
    note/.style={
        draw, dashed, rounded corners, inner sep=6pt, align=left
    },
    arrow/.style={-{Latex[length=3mm]}, thick},
    lbl/.style={midway, below=2pt, align=center, font=\scriptsize},
    lbltop/.style={midway, above=2pt, align=center, font=\scriptsize}
]

\node[smallblock] (l1) {Lane 1\\$k_1,c_1$};
\node[smallblock, below=0.5cm of l1]  (l2) {Lane 2\\$k_2,c_2$};
\node[smallblock, below=0.5cm of l2]  (l3) {Lane 3\\$k_3,c_3$};
\node[smallblock, below=0.5cm of l3]  (l4) {Lane 4\\$k_4,c_4$};

\node[smallblock, right=2.4cm of l1]  (n12)   {$=$ ?\\\footnotesize merge/add};
\node[smallblock, right=2.4cm of l3]  (n34)   {$=$ ?\\\footnotesize merge/add};
\node[smallblock, right=2.6cm of $(n12)!0.5!(n34)$] (n1234)
      {$=$ ?\\\footnotesize merge/add};

\node[draw=none, above=0.5cm of l1]    (laneslabel) {\textbf{Input Lanes}};
\node[draw=none, above=0.5cm of n12]   (crnlabel)
      {\textbf{Conflict Resolution Network (CRN)}};

\node[block, right=3.2cm of n1234] (hist)
      {Histogram Memory\\$H[k]$\\\textbf{canonical ordered state}};

\node[block, right=3.2cm of hist]  (scan)
      {Geometric Scan\\$k=0,\dots,U{-}1$};

\node[block, right=3.0cm of scan]  (out)
      {Output Vector\\\textit{optional projection}};

\draw[arrow] (l1.east) -- (n12.west);
\draw[arrow] (l2.east) -- (n12.west);
\draw[arrow] (l3.east) -- (n34.west);
\draw[arrow] (l4.east) -- (n34.west);
\draw[arrow] (n12.east) -- (n1234.west);
\draw[arrow] (n34.east) -- (n1234.west);
\draw[arrow] (n1234.east) --
    node[lbl] {additive writes,\\no order comparison} (hist.west);
\draw[arrow] (hist.east) --
    node[lbltop] {state read} (scan.west);
\draw[arrow] (scan.east) --
    node[lbl] {emit $k$ exactly\\$H[k]$ times} (out.west);

\node[note, fit=(l1)(l2)(l3)(l4)(laneslabel)] {};
\node[note, fit=(n12)(n34)(n1234)(crnlabel)]  {};

\node[note, below=1.2cm of hist, text width=5.0cm] (notehist) {
    \textbf{Key claim:} $H[k]$ is not an intermediate structure.
    It is the complete implicit ordered representation of the dataset.
    No prefix sum is required.
};
\node[note, below=1.2cm of out, text width=3.6cm] (noteout) {
    \textbf{Output view:} the output vector exists only when
    a linear array interface is required.
};
\end{tikzpicture}%
}
\caption{
\dialsort{} conceptual pipeline with CRN integrated.
Input lanes ingest key-count pairs in parallel.
The CRN merges equal keys additively using equality detection only
($=$ check; no $<$ or $>$), producing conflict-free updates to
histogram memory $H[k]$.
$H[\cdot]$ is the canonical ordered state;
the geometric scan and output vector are derived interface operations,
not the core ordering mechanism.
}
\label{fig:dialsort_crn_pipeline}
\end{figure*}

Unlike counting sort, \dialsort{} does not treat the histogram as a
temporary counting structure preceding reconstruction.
In \dialsort{}, the histogram \emph{is} the ordered representation;
reconstruction is optional and interface-dependent.

\section{Complexity Analysis\label{sec:complexity}}

The reviewer correctly notes that numerical examples illustrate but do
not prove complexity.  We provide formal proofs below.  Complexity is
analysed in the RAM model; the absence of order comparisons refers to
the comparison-based view of ordering established in
Definition~\ref{def:ordercomp}.

\begin{theorem}[Sequential complexity\label{thm:seq}]
\dialsort{} sorts $n$ integer keys from $[0,U{-}1]$ in
$\bigO{n+U}$ time and $\bigO{U}$ auxiliary space.
\end{theorem}

\begin{proof}
Initialising $H[0..U{-}1]$ to zero costs $\Theta(U)$.  Ingestion
(Eq.~\ref{eq:ingestion}) performs one constant-time indexed increment
per input key, for total cost $O(n)$ in the RAM model.  The ordered
projection (Eq.~\ref{eq:projection}) visits all $U$ histogram cells
and emits exactly $n = \sum_k H[k]$ values overall, so its cost is
$O(U+n)$.  Therefore the overall time complexity is $O(n+U)$, and the
auxiliary space complexity is $O(U)$.
\end{proof}

\begin{corollary}\label{cor:linear}
When $U = O(n)$, \dialsort{} runs in $O(n)$ time.
\end{corollary}

\begin{remark}[Practical domain guideline]
The condition $U = O(n)$ is a theoretical bound.  As an empirical
engineering heuristic, \dialsort{} delivers its strongest gains when
$n \gtrsim 10 \cdot U$; when $U \gg n$, the $O(U)$ scan cost
dominates and alternative methods should be preferred.  This
guideline is revisited in the Limitations section
(Section~\ref{sec:limits}).
\end{remark}

\begin{theorem}[Parallel \dialsort{} with pipelined CRN\label{thm:par}]
Assume: (i)~$k$ ingestion lanes; (ii)~a pipelined CRN of depth
$L=\lceil \log_2 k \rceil$; (iii)~one pipeline advance per cycle;
(iv)~a histogram memory interface capable of accepting the
conflict-free CRN outputs at the sustained rate of the pipeline model;
(v)~equality tests and indexed histogram updates are $O(1)$ operations
in the underlying word-RAM or hardware model.

Then the ingestion phase completes in
\[
T_{\mathrm{ingest}} = O\!\left(\left\lceil \tfrac{n}{k}\right\rceil
+ \lceil \log_2 k \rceil\right),
\]
the subsequent ordered scan costs $T_{\mathrm{scan}} = O(U)$, giving
total wall-clock time
\[
T_{\mathrm{total}} = T_{\mathrm{ingest}} + T_{\mathrm{scan}}
= O\!\left(\tfrac{n}{k} + \log k + U\right),
\]
using $O(U+k)$ auxiliary space.
\end{theorem}

\begin{proof}
Distribute the $n$ keys across $k$ ingestion lanes; each lane processes
at most $\lceil n/k \rceil$ keys.  The CRN is a pipeline of depth
$L = \lceil \log_2 k \rceil$, so after the first inputs enter the
pipeline, a full traversal requires $L$ cycles of latency.  During the first $L$ cycles the pipeline is filling and does not yet
produce outputs at full rate; this initial fill cost is exactly $O(L)
= O(\log k)$, already absorbed in the latency term.  Once the pipeline
is full (after $L$ cycles), it sustains one advance per cycle in steady
state, processing one batch of up to $k$ keys per cycle.  Therefore
$T_{\mathrm{ingest}} = O(\lceil n/k\rceil + L) = O(n/k + \log k)$.

By Lemma~\ref{lem:crn}, all equal-key contributions reaching the CRN
in the same cycle are additively merged without loss, so ingestion
correctness is preserved under pipelining.  After ingestion completes,
$H$ contains the correct multiplicities (Theorem~\ref{thm:correct}).
The ordered projection scans $U$ cells once: $T_{\mathrm{scan}} = O(U)$.
For space: $O(U)$ for $H[\cdot]$ plus $O(k)$ for CRN pipeline buffers.
\end{proof}

\begin{remark}
The $\log k$ term is a pipeline-latency term, not a per-key
multiplicative cost.  $T_{\mathrm{ingest}}$ and $T_{\mathrm{scan}}$
are stated separately to make it explicit that they do not overlap
in the sequential pipeline implementation (ingestion completes before
the scan begins); overlap is possible in streaming implementations
but requires additional synchronisation not modelled here.
\end{remark}

\section{System Visualization}

Figure~\ref{fig:viz} provides a cycle-accurate visualization of
\dialsort{} for a small input, illustrating the three phases:
key ingestion (one \texttt{H[k]++} per cycle, no comparisons),
the canonical histogram state after ingestion, and the optional
output projection.

\begin{figure*}[tp]
\centering
\resizebox{\textwidth}{!}{%
\begin{tikzpicture}[
    font=\small,
    arrow/.style={-{Latex[length=3mm]}, thick},
    lbl/.style={font=\scriptsize, align=center},
    note/.style={draw, dashed, rounded corners, inner sep=5pt, align=left},
]

\node[draw=none, font=\small\bfseries] at (1.8,  7.7) {Input keys $K$};
\node[draw=none, font=\small\bfseries] at (7.5,  7.7) {Cycle-accurate ingestion};
\node[draw=none, font=\small\bfseries] at (14.0, 7.7) {Histogram $H[k]$ after ingestion};
\node[draw=none, font=\scriptsize]     at (14.0, 7.25)
      {(canonical ordered state --- no prefix sum needed)};
\node[draw=none, font=\small\bfseries] at (20.5, 7.7) {Output vector (optional)};
\node[draw=none, font=\scriptsize]     at (20.5, 7.25) {projection only if needed};

\foreach \val/\row in {3/0, 1/1, 3/2, 5/3, 1/4, 3/5, 5/6, 3/7}{
    \draw[fill=gray!15] (0.9, 6.2-\row*0.72) rectangle (2.7, 6.2-\row*0.72+0.62);
    \node[lbl] at (1.8, 6.2-\row*0.72+0.31) {$k=\val$};
}
\node[lbl] at (1.8, 0.5) {$\vdots$};

\draw[fill=gray!30] (3.8,6.82) rectangle (11.2,7.22);
\node[lbl,font=\footnotesize\bfseries] at (5.2, 7.02) {Cycle};
\node[lbl,font=\footnotesize\bfseries] at (7.5, 7.02) {Key arrives};
\node[lbl,font=\footnotesize\bfseries] at (9.8, 7.02) {Operation};

\foreach \cyc/\key/\op/\row in {
    1/3/{H[3]++}/0,
    2/1/{H[1]++}/1,
    3/3/{H[3]++}/2,
    4/5/{H[5]++}/3,
    5/1/{H[1]++}/4,
    6/3/{H[3]++}/5,
    7/5/{H[5]++}/6,
    8/3/{H[3]++}/7
}{
    \pgfmathsetmacro{\ybot}{6.2-\row*0.72}
    \pgfmathsetmacro{\ytop}{\ybot+0.62}
    \pgfmathparse{mod(\row,2)==0 ? "gray!8" : "white"}
    \edef\fillcol{\pgfmathresult}
    \draw[fill=\fillcol] (3.8,\ybot) rectangle (11.2,\ytop);
    \node[lbl] at (5.2,  \ybot+0.31) {\cyc};
    \node[lbl] at (7.5,  \ybot+0.31) {$k=\key$};
    \node[lbl,font=\ttfamily\scriptsize] at (9.9, \ybot+0.31) {\op};
}
\draw[thick] (3.8,0.82) rectangle (11.2,7.22);
\draw[thick] (3.8,6.82) -- (11.2,6.82);
\draw[thick] (6.2,0.82) -- (6.2,7.22);
\draw[thick] (8.8,0.82) -- (8.8,7.22);
\node[lbl] at (7.5, 0.5) {No comparison at any cycle};

\draw[thick] (11.8, 0.8) -- (16.2, 0.8);
\draw[thick] (11.8, 0.8) -- (11.8, 6.8);
\node[lbl,rotate=90] at (11.3, 3.8) {$H[k]$};
\node[lbl] at (14.0, 0.3) {key $k$};

\draw[fill=gray!10] (12.0,0.8) rectangle (12.5,0.85);
\node[lbl,font=\tiny] at (12.25,0.5) {$k{=}0$};
\node[lbl,font=\tiny] at (12.25,1.0) {0};
\draw[fill=gray!60] (12.6,0.8) rectangle (13.1,2.3);
\node[lbl,font=\tiny] at (12.85,0.5) {$k{=}1$};
\node[lbl,font=\tiny] at (12.85,2.5) {2};
\draw[fill=gray!10] (13.2,0.8) rectangle (13.7,0.85);
\node[lbl,font=\tiny] at (13.45,0.5) {$k{=}2$};
\node[lbl,font=\tiny] at (13.45,1.0) {0};
\draw[fill=gray!80] (13.8,0.8) rectangle (14.3,5.8);
\node[lbl,font=\tiny] at (14.05,0.5) {$k{=}3$};
\node[lbl,font=\tiny] at (14.05,6.05) {4};
\draw[fill=gray!10] (14.4,0.8) rectangle (14.9,0.85);
\node[lbl,font=\tiny] at (14.65,0.5) {$k{=}4$};
\node[lbl,font=\tiny] at (14.65,1.0) {0};
\draw[fill=gray!60] (15.0,0.8) rectangle (15.5,2.3);
\node[lbl,font=\tiny] at (15.25,0.5) {$k{=}5$};
\node[lbl,font=\tiny] at (15.25,2.5) {2};
\node[lbl] at (16.0,0.5) {$\cdots$};

\node[note, text width=3.8cm] at (14.0, -0.8) {
    \textbf{H[k] is complete.}\\
    Sorted order exists\\
    in the index space.\\
    No prefix sum needed.\\
    No reconstruction yet.
};

\draw[arrow,thick] (11.25, 3.8) -- (11.75, 3.8);
\node[lbl,font=\scriptsize] at (11.0, 2.9) {after\\ingestion};

\foreach \val/\row in {1/0, 1/1, 3/2, 3/3, 3/4, 3/5, 5/6, 5/7}{
    \pgfmathparse{\val==1 ? "gray!25" : (\val==3 ? "gray!80" : "gray!50")}
    \edef\fillcol{\pgfmathresult}
    \draw[fill=\fillcol] (19.2, 6.2-\row*0.72) rectangle (21.8, 6.2-\row*0.72+0.62);
    \node[lbl] at (20.5, 6.2-\row*0.72+0.31) {$k=\val$};
}

\node[note, text width=3.5cm] at (20.5,-0.8) {
    \textbf{Optional.}\\
    Only emitted when\\
    a contiguous array\\
    interface is required.\\
    Cost: $O(U + n)$,\\
    zero comparisons.
};

\draw[arrow,thick] (16.3, 3.8) -- (19.1, 3.8);
\node[lbl,font=\scriptsize,text width=2.5cm] at (17.7, 4.4)
    {scan $k=0\ldots U{-}1$\\emit $H[k]$ times};
\node[draw,dashed,rounded corners,inner sep=3pt,font=\scriptsize]
    at (17.7, 2.6) {no comparison};

\draw[dashed,gray!60] (3.5,  -1.5) -- (3.5,  7.4);
\draw[dashed,gray!60] (11.5, -1.5) -- (11.5, 7.4);
\draw[dashed,gray!60] (16.5, -1.5) -- (16.5, 7.4);

\end{tikzpicture}%
}
\caption{
Cycle-accurate \dialsort{} visualization for input
$K = \{3,1,3,5,1,3,5,3\}$ ($n{=}8$, $U{=}6$).
\textbf{Left:} the arriving key stream.
\textbf{Centre-left:} each cycle performs exactly one \texttt{H[k]++}
--- no order comparison is ever evaluated.
\textbf{Centre-right:} the histogram $H[\cdot]$ after all 8 cycles;
this is the canonical ordered representation --- no prefix sum, no
reconstruction required.
\textbf{Right:} the optional output vector, obtained by scanning
$H[0\ldots U{-}1]$ and emitting $k$ exactly $H[k]$ times; this step
is executed only when a contiguous array interface is required.
}
\label{fig:viz}
\end{figure*}
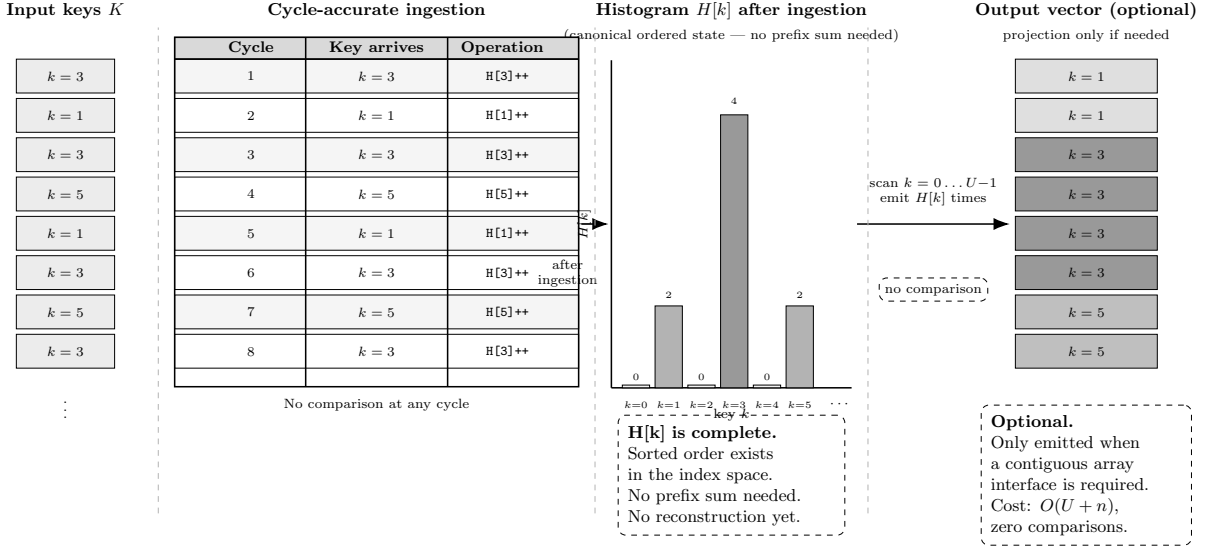

\section{Conflict Resolution Network (CRN)}

\subsection{The conflict problem}

When $k$ lanes ingest keys in parallel, two or more lanes may carry the
same key within the same clock cycle.  Naive solutions serialize writes
(losing parallelism) or use atomic increments (hardware serialization).
The CRN resolves conflicts before writes reach the register bank,
without serialization and without order comparisons.

\subsection{CRN Lane States, Node Definition, and Level Transition\label{sec:crn_def}}

In the \emph{base formulation}, each lane carries a pair $(k_i, 1)$ per
ingested key --- one key, one unit count.  A count $c_i > 1$ arises
within a lane when the same key appears in consecutive positions; in
that case the lane may accumulate locally before passing to the CRN.
This local accumulation is an \emph{optimisation}, not a requirement:
correctness holds for $c_i = 1$ in all cases by Lemma~\ref{lem:crn}.

At any clock cycle, lane $i$ formally carries a pair $(k_i,\, c_i)$
where:
\begin{itemize}
  \item $k_i \in [0,U{-}1]$ is the \emph{key value} currently present
        in lane $i$;
  \item $c_i \in \mathbb{Z}_{\geq 1}$ is the \emph{count} contributed
        by lane $i$ at this cycle (equal to~1 in the base formulation,
        or the locally accumulated total when the optimisation is
        active).
\end{itemize}

Each CRN node receives two lane pairs and computes:
\begin{equation}
k_i = k_j \;\Rightarrow\; \text{emit }\bigl(k_i,\; c_i{+}c_j\bigr)
\;\text{(to next level or }H[k_i]\text{)}
\label{eq:crn_eq}
\end{equation}
\vspace{-4pt}
\begin{equation}
k_i \neq k_j \;\Rightarrow\; \text{forward }(k_i,c_i),\,(k_j,c_j)
\;\text{to next level}
\label{eq:crn_neq}
\end{equation}

Only equality is tested --- a bitwise XOR-reduce on the binary
representation of the keys.  No order comparison
(Definition~\ref{def:ordercomp}) is evaluated.
\emph{Emit to next level} applies at internal CRN levels; \emph{emit
to histogram memory} applies only at the final CRN output.

\begin{definition}[CRN level transition\label{def:crn_level}]
Fix a cycle $t$.  Let $L_t^{(0)}$ denote the multiset of lane states
entering the CRN at that cycle, where each state is a pair $(a,c)$
with key $a \in [0,U-1]$ and count $c \in \mathbb{N}_{\geq 1}$.

For each level $\ell \geq 0$, define $L_t^{(\ell+1)}$ from
$L_t^{(\ell)}$ by grouping adjacent items into consecutive pairs.
For each pair:
\begin{enumerate}
  \item if both keys are equal, replace $(a,c_1),(a,c_2)$ by the
        single item $(a,\,c_1+c_2)$;
  \item if the keys are distinct, forward both items unchanged;
  \item if one unpaired residual item remains (when the number of
        items is odd), forward it unchanged to the next level.
\end{enumerate}

After $L = \lceil \log_2 k \rceil$ levels, every key appears at most
once in $L_t^{(L)}$, and each surviving pair is written exactly once
to histogram memory.
\end{definition}

\subsection{CRN vs.\ existing parallel primitives}

The CRN is structurally distinct from three related hardware primitives:

\textbf{Bitonic sorting networks}~\cite{batcher1968} sort $n$ elements
using $O(n \log^2 n)$ comparators arranged in a fixed pattern.  Every
comparator node evaluates a \emph{magnitude comparison} ($>$) and
conditionally swaps two elements.  The CRN evaluates only
\emph{equality} ($=$) and never reorders elements --- it fuses concurrent
counts of identical keys.

\textbf{Parallel prefix-sum (scan) trees}~\cite{cuda_scan} compute
cumulative sums over an array using a balanced reduction tree, requiring
$O(\log n)$ stages and $O(n)$ operations.  Counting sort uses a
prefix-sum pass to convert the histogram into output positions.
The CRN requires no prefix sum: it operates during ingestion, before
the histogram is fully populated, and does not need cumulative position
information.

\textbf{GPU histogram parallelization} uses atomic operations or
privatized per-thread histograms with a merge step to handle write
conflicts.  The CRN provides a formal architectural primitive that
achieves the same result through a structured reduction tree with
bounded, clock-deterministic latency ($\lceil \log_2 k \rceil$ cycles
for $k$ lanes), independent of collision rate --- a property not
achievable with atomic-based schemes.

\subsection{Additive invariant}

The CRN is an \emph{additive reduction, not an arbitration}.  Three
simultaneous arrivals of key 42 produce exactly one write
$H[42] \mathrel{+}= 3$.  Nothing is dropped; the histogram receives the
correct total increment.  Table~\ref{tab:crn} summarizes the invariant.

\begin{lemma}[CRN additive conservation\label{lem:crn}]
For every cycle $t$, every level $\ell$, and every key $a$, the total
count associated with key $a$ is invariant across CRN levels:
\[
\sum_{\substack{(k,c)\in L_t^{(\ell)} \\ k=a}} c
\;=\;
\sum_{\substack{(k,c)\in L_t^{(\ell+1)} \\ k=a}} c.
\]
Consequently, after the final level $L = \lceil\log_2 k\rceil$, each
key $a$ appears at most once in $L_t^{(L)}$ and carries exactly the
total multiplicity contributed at cycle $t$.
\end{lemma}

\begin{proof}
Fix a cycle $t$ and a key $a$.  At any CRN level, adjacent items are
processed pairwise by Definition~\ref{def:crn_level}.
\emph{Case~1:} two items with key $a$ are merged --- their counts are
replaced by a single item with count equal to their sum, so the total
contribution of key $a$ is preserved.
\emph{Case~2:} two items have distinct keys --- both are forwarded
unchanged, preserving the total contribution of key $a$.
\emph{Case~3:} an unpaired residual item is forwarded unchanged,
again preserving the total.

In all three cases the total count for key $a$ is invariant from
level $\ell$ to level $\ell+1$.  By induction over all
$L = \lceil\log_2 k\rceil$ levels the invariant holds throughout the
entire CRN.  After the final level, all same-key contributions that
can be merged at cycle $t$ have been additively collapsed, so each
surviving key carries exactly the total multiplicity contributed at
that cycle.
\end{proof}

\begin{table}[htbp]
\centering
\caption{CRN invariant: what is fused and what is preserved.}
\label{tab:crn}
\small
\begin{tabular}{@{}lp{2.0cm}p{2.3cm}@{}}
\toprule
\textbf{Property} & \textbf{Fused} & \textbf{Preserved} \\
\midrule
Concurrent writes   & $m$ events $\to$ 1  & Total count ($+m$) \\
Key identity        & Which lane           & Aggregate frequency \\
Write-port conflict & Register contention  & Histogram semantics \\
Semantic model      & Arbitration          & Additive correctness \\
\bottomrule
\end{tabular}
\end{table}

\subsection{Performance consequence under skewed input}

Under skewed distributions (80\% of keys in 5\% of the universe), most
lanes in the same clock cycle carry the same key.  The CRN collapses
these into single additive writes, dramatically reducing write traffic.
This structural mechanism produces the strongest speedups observed
experimentally (Section~\ref{sec:results}).

\subsection{Pipeline latency}

For $k{=}16$ lanes: $\lceil \log_2 16 \rceil = 4$ stages.  CRN latency
is 4 cycles at any clock frequency, independent of $n$.  This bounded
overhead is not achievable with arbitration-based schemes.

\begin{simbox}
\small\textbf{Simulator S-CRN --- Conflict Resolution Network}\\
\url{https://elmaestrotic.github.io/dsort/simulators/crn/}\\[2pt]
Interactive step-through of the CRN pipeline: set the number of lanes
$k$, inject key streams, and observe additive merging across each level
$\ell = 0, \ldots, \lceil\log_2 k\rceil$.  Demonstrates the invariant
of Lemma~\ref{lem:crn} and the handling of odd-lane residuals
(Definition~\ref{def:crn_level}).
\end{simbox}

\section{Physical Interpretation and Substrate Analysis\label{sec:substrates}}

\begin{quote}\small
\textbf{Terminology note.}  This version renames ``Tiers'' to
\emph{Substrates} to avoid ambiguity with data-centre resilience
terminology (where ``Tier'' refers to levels of redundancy and
uptime).  The four substrates --- CPU software, FPGA, photonic, and
WSE-3 --- correspond to distinct physical execution paradigms.
\end{quote}

\begin{quote}\small
\textbf{Note on simulators.}  The interactive simulators accompanying
each substrate constitute \emph{pedagogical illustrations} of the
operational principle of that substrate.  They are not experimental
validation.  Source code is released at the URLs below; links will
be updated to permanent GitHub Pages addresses upon publication.
\end{quote}

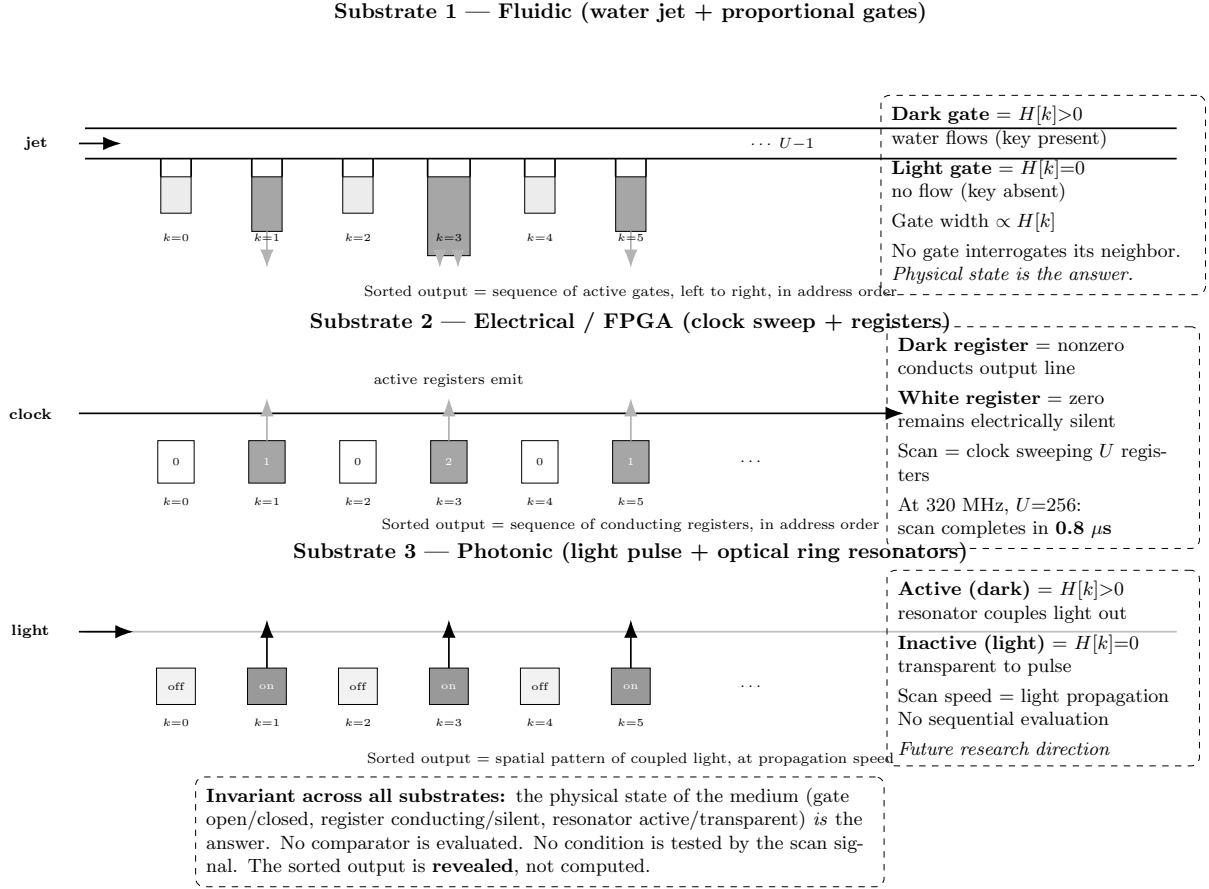
\begin{figure*}[tp]
\centering
\resizebox{\textwidth}{!}{%
\begin{tikzpicture}[
    font=\small,
    arrow/.style={-{Latex[length=3mm]}, thick},
    lbl/.style={font=\scriptsize, align=center},
    note/.style={draw, dashed, rounded corners, inner sep=5pt, align=left},
]

\node[draw=none, font=\normalsize\bfseries] at (9, 8.4)
      {Substrate 1 \textemdash\ Fluidic (water jet + proportional gates)};

\draw[thick] (0,6.5) -- (18,6.5);
\draw[thick] (0,6.0) -- (18,6.0);
\node[lbl]   at (-0.8,6.25) {\textbf{jet}};
\draw[arrow,thick] (-0.1,6.25) -- (0.6,6.25);

\draw[fill=gray!15] (1.25,5.1) rectangle (1.75,5.7);
\draw[thick](1.25,6.0)--(1.25,5.7); \draw[thick](1.75,6.0)--(1.75,5.7);
\node[lbl] at (1.5,4.7) {\tiny $k{=}0$};
\draw[fill=gray!70] (2.75,4.8) rectangle (3.25,5.7);
\draw[thick](2.75,6.0)--(2.75,5.7); \draw[thick](3.25,6.0)--(3.25,5.7);
\draw[arrow,gray!60](3.0,4.8)--(3.0,4.2);
\node[lbl] at (3.0,4.7) {\tiny $k{=}1$};
\draw[fill=gray!15] (4.25,5.1) rectangle (4.75,5.7);
\draw[thick](4.25,6.0)--(4.25,5.7); \draw[thick](4.75,6.0)--(4.75,5.7);
\node[lbl] at (4.5,4.7) {\tiny $k{=}2$};
\draw[fill=gray!70] (5.65,4.4) rectangle (6.35,5.7);
\draw[thick](5.65,6.0)--(5.65,5.7); \draw[thick](6.35,6.0)--(6.35,5.7);
\draw[arrow,gray!60](5.85,4.4)--(5.85,4.2);
\draw[arrow,gray!60](6.15,4.4)--(6.15,4.2);
\node[lbl] at (6.0,4.7) {\tiny $k{=}3$};
\draw[fill=gray!15] (7.25,5.1) rectangle (7.75,5.7);
\draw[thick](7.25,6.0)--(7.25,5.7); \draw[thick](7.75,6.0)--(7.75,5.7);
\node[lbl] at (7.5,4.7) {\tiny $k{=}4$};
\draw[fill=gray!70] (8.75,4.8) rectangle (9.25,5.7);
\draw[thick](8.75,6.0)--(8.75,5.7); \draw[thick](9.25,6.0)--(9.25,5.7);
\draw[arrow,gray!60](9.0,4.8)--(9.0,4.2);
\node[lbl] at (9.0,4.7) {\tiny $k{=}5$};
\node[lbl] at (11.5,6.25) {$\cdots\;U{-}1$};
\node[note, text width=5.0cm] at (15.8,5.4) {
    \textbf{Dark gate} = $H[k]{>}0$\\water flows (key present)\\[3pt]
    \textbf{Light gate} = $H[k]{=}0$\\no flow (key absent)\\[3pt]
    Gate width $\propto H[k]$\\[3pt]
    No gate interrogates its neighbor.\\\textit{Physical state is the answer.}
};
\node[lbl] at (9,3.8)
    {Sorted output = sequence of active gates, left to right, in address order};

\node[draw=none, font=\normalsize\bfseries] at (9,3.3)
      {Substrate 2 \textemdash\ Electrical / FPGA (clock sweep + registers)};
\draw[arrow,thick] (-0.1,1.8) -- (13.5,1.8);
\node[lbl] at (-0.9,1.8) {\textbf{clock}};
\draw[fill=white] (1.2,0.65) rectangle (1.8,1.35); \node[lbl,font=\tiny] at (1.5,1.0) {0}; \node[lbl] at (1.5,0.35) {\tiny $k{=}0$};
\draw[fill=gray!70] (2.7,0.65) rectangle (3.3,1.35); \node[lbl,font=\tiny,text=white] at (3.0,1.0) {1}; \node[lbl] at (3.0,0.35) {\tiny $k{=}1$};
\draw[arrow,gray!60,thick](3.0,1.35)--(3.0,2.05);
\draw[fill=white] (4.2,0.65) rectangle (4.8,1.35); \node[lbl,font=\tiny] at (4.5,1.0) {0}; \node[lbl] at (4.5,0.35) {\tiny $k{=}2$};
\draw[fill=gray!70] (5.7,0.65) rectangle (6.3,1.35); \node[lbl,font=\tiny,text=white] at (6.0,1.0) {2}; \node[lbl] at (6.0,0.35) {\tiny $k{=}3$};
\draw[arrow,gray!60,thick](6.0,1.35)--(6.0,2.05);
\draw[fill=white] (7.2,0.65) rectangle (7.8,1.35); \node[lbl,font=\tiny] at (7.5,1.0) {0}; \node[lbl] at (7.5,0.35) {\tiny $k{=}4$};
\draw[fill=gray!70] (8.7,0.65) rectangle (9.3,1.35); \node[lbl,font=\tiny,text=white] at (9.0,1.0) {1}; \node[lbl] at (9.0,0.35) {\tiny $k{=}5$};
\draw[arrow,gray!60,thick](9.0,1.35)--(9.0,2.05);
\node[lbl] at (11.0,1.0) {$\cdots$};
\node[lbl] at (6.0,2.35) {\scriptsize active registers emit};
\node[note, text width=4.8cm] at (15.8,1.4) {
    \textbf{Dark register} = nonzero\\conducts output line\\[3pt]
    \textbf{White register} = zero\\remains electrically silent\\[3pt]
    Scan = clock sweeping $U$ registers\\[3pt]
    At 320~MHz, $U{=}256$:\\scan completes in \textbf{0.8~$\mu$s}
};
\node[lbl] at (9,-0.05)
    {Sorted output = sequence of conducting registers, in address order};

\node[draw=none, font=\normalsize\bfseries] at (9,-0.5)
      {Substrate 3 \textemdash\ Photonic (light pulse + optical ring resonators)};
\draw[thick,gray!50] (0,-1.8) -- (18,-1.8);
\node[lbl] at (-0.9,-1.8) {\textbf{light}};
\draw[arrow,thick] (-0.1,-1.8) -- (0.8,-1.8);
\draw[fill=gray!10] (1.18,-3.0) rectangle (1.82,-2.4); \node[lbl,font=\tiny] at (1.5,-2.7) {off}; \node[lbl] at (1.5,-3.3) {\tiny $k{=}0$};
\draw[fill=gray!80] (2.68,-3.0) rectangle (3.32,-2.4); \node[lbl,font=\tiny,text=white] at (3.0,-2.7) {on}; \node[lbl] at (3.0,-3.3) {\tiny $k{=}1$};
\draw[arrow,thick](3.0,-2.4)--(3.0,-1.6);
\draw[fill=gray!10] (4.18,-3.0) rectangle (4.82,-2.4); \node[lbl,font=\tiny] at (4.5,-2.7) {off}; \node[lbl] at (4.5,-3.3) {\tiny $k{=}2$};
\draw[fill=gray!80] (5.68,-3.0) rectangle (6.32,-2.4); \node[lbl,font=\tiny,text=white] at (6.0,-2.7) {on}; \node[lbl] at (6.0,-3.3) {\tiny $k{=}3$};
\draw[arrow,thick](6.0,-2.4)--(6.0,-1.6);
\draw[fill=gray!10] (7.18,-3.0) rectangle (7.82,-2.4); \node[lbl,font=\tiny] at (7.5,-2.7) {off}; \node[lbl] at (7.5,-3.3) {\tiny $k{=}4$};
\draw[fill=gray!80] (8.68,-3.0) rectangle (9.32,-2.4); \node[lbl,font=\tiny,text=white] at (9.0,-2.7) {on}; \node[lbl] at (9.0,-3.3) {\tiny $k{=}5$};
\draw[arrow,thick](9.0,-2.4)--(9.0,-1.6);
\node[lbl] at (11.0,-2.7) {$\cdots$};
\node[note, text width=4.8cm] at (15.8,-2.4) {
    \textbf{Active (dark)} = $H[k]{>}0$\\resonator couples light out\\[3pt]
    \textbf{Inactive (light)} = $H[k]{=}0$\\transparent to pulse\\[3pt]
    Scan speed = light propagation\\No sequential evaluation\\[3pt]\textit{Future research direction}
};
\node[lbl] at (9,-3.9)
    {Sorted output = spatial pattern of coupled light, at propagation speed};

\node[note, text width=11cm] at (7.5,-5.1) {
    \textbf{Invariant across all substrates:}
    the physical state of the medium
    (gate open/closed, register conducting/silent,
    resonator active/transparent)
    \emph{is} the answer.
    No comparator is evaluated. No condition is tested by the scan signal.
    The sorted output is \textbf{revealed}, not computed.
};
\end{tikzpicture}%
}
\caption{
\dialsort{} across three physical substrates, shown for histogram state
$H{=}\{k{=}1{:}1,\;k{=}3{:}2,\;k{=}5{:}1\}$.
\textbf{Top (Substrate~1):} water jet; dark gates admit flow $\propto H[k]$.
\textbf{Middle (Substrate~2):} clock sweeps $U$ FPGA registers; nonzero
registers conduct.
\textbf{Bottom (Substrate~3):} light pulse; active resonators couple light out.
Substrate~2 is an analytical projection; Substrate~3 is a future
research direction.
}
\label{fig:physical}
\end{figure*}

The most intuitive reading of \dialsort{} is the
\emph{abacus-and-torch analogy}.
Imagine an abacus whose beads have already settled into position after
ingestion: column $k$ holds exactly $H[k]$ beads, stacked from the
bottom rail upward.
No bead consulted any neighbour during this process --- each column is
sovereign.
A torch now sweeps from column $0$ to column $U{-}1$.
Its beam does not decide, compare, or rearrange anything; it simply
illuminates what is already there.
The sorted output is the shadow projected on the wall behind the abacus
--- taller shadow, more beads, higher frequency.
The order was encoded in the abacus from the moment of ingestion.

Figure~\ref{fig:abacus} illustrates this analogy for a small example
$H = \{k{=}1{:}2,\;k{=}2{:}1,\;k{=}3{:}4,\;k{=}5{:}3\}$.

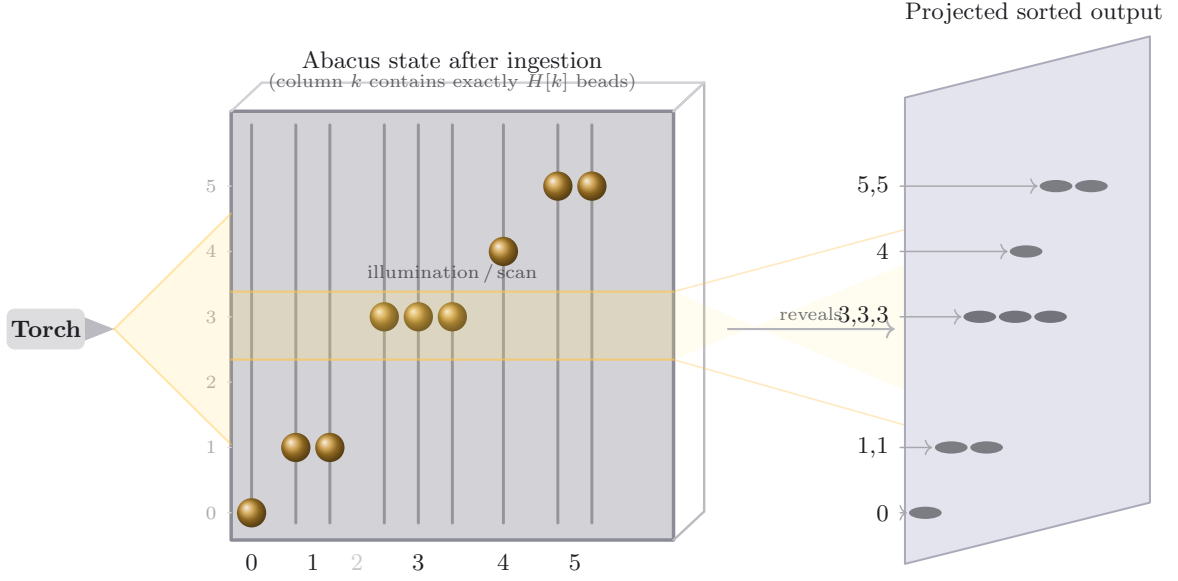
\begin{figure*}[t]
\centering
\begin{tikzpicture}[scale=0.90, line join=round, line cap=round]

  \definecolor{framefill}{RGB}{210,210,215}
  \definecolor{framec}{RGB}{140,140,150}
  \definecolor{rodc}{RGB}{150,148,148}
  \definecolor{beadc}{RGB}{200,145,45}
  \definecolor{wallc}{RGB}{228,228,238}
  \definecolor{wallborder}{RGB}{170,170,185}
  \definecolor{shadowc}{RGB}{80,80,90}
  \definecolor{lightc}{RGB}{255,225,110}
  \definecolor{textc}{RGB}{30,30,35}

  \fill[framefill] (1.0,0.80) rectangle (7.5,7.10);
  \draw[framec, line width=1.3pt] (1.0,0.80) rectangle (7.5,7.10);
  \draw[framec!60, line width=0.9pt] (7.5,0.80) -- (7.95,1.22);
  \draw[framec!60, line width=0.9pt] (7.5,7.10) -- (7.95,7.52);
  \draw[framec!60, line width=0.9pt] (7.95,1.22) -- (7.95,7.52);
  \draw[framec!40, line width=0.9pt] (1.0, 7.10) -- (1.45,7.52);
  \draw[framec!40, line width=0.9pt] (7.5, 7.10) -- (7.95,7.52);
  \draw[framec!40, line width=0.9pt] (1.45,7.52) -- (7.95,7.52);

  \node[font=\small, align=center, text=textc] at (4.25,7.82)
    {Abacus state after ingestion};
  \node[font=\scriptsize, align=center, text=textc!70] at (4.25,7.52)
    {(column $k$ contains exactly $H[k]$ beads)};

  \foreach \k in {0,1,2,3,4,5} {
    \pgfmathsetmacro{\yval}{1.20 + \k * 0.96}
    \node[font=\scriptsize, anchor=east, text=rodc!80] at (0.92,\yval) {$\k$};
    \draw[rodc!35, line width=0.5pt] (0.95,\yval) -- (1.0,\yval);
  }

  \draw[rodc, line width=1.1pt] (1.30,1.05) -- (1.30,6.90);
  \draw[rodc, line width=1.1pt] (1.95,1.05) -- (1.95,6.90);
  \draw[rodc, line width=1.1pt] (2.45,1.05) -- (2.45,6.90);
  \draw[rodc, line width=1.1pt] (3.25,1.05) -- (3.25,6.90);
  \draw[rodc, line width=1.1pt] (3.75,1.05) -- (3.75,6.90);
  \draw[rodc, line width=1.1pt] (4.25,1.05) -- (4.25,6.90);
  \draw[rodc, line width=1.1pt] (5.00,1.05) -- (5.00,6.90);
  \draw[rodc, line width=1.1pt] (5.80,1.05) -- (5.80,6.90);
  \draw[rodc, line width=1.1pt] (6.30,1.05) -- (6.30,6.90);

  \node[font=\small, text=textc] at (1.30,0.47) {$0$};
  \node[font=\small, text=textc] at (2.20,0.47) {$1$};
  \node[font=\small, text=rodc!50] at (2.85,0.47) {$2$};
  \node[font=\small, text=textc] at (3.75,0.47) {$3$};
  \node[font=\small, text=textc] at (5.00,0.47) {$4$};
  \node[font=\small, text=textc] at (6.05,0.47) {$5$};


  \shade[ball color=beadc] (1.30,1.20) circle (0.215);

  \shade[ball color=beadc] (1.95,2.16) circle (0.215);
  \shade[ball color=beadc] (2.45,2.16) circle (0.215);


  \shade[ball color=beadc] (3.25,4.08) circle (0.215);
  \shade[ball color=beadc] (3.75,4.08) circle (0.215);
  \shade[ball color=beadc] (4.25,4.08) circle (0.215);

  \shade[ball color=beadc] (5.00,5.04) circle (0.215);

  \shade[ball color=beadc] (5.80,6.00) circle (0.215);
  \shade[ball color=beadc] (6.30,6.00) circle (0.215);

  \fill[framec!30, rounded corners=3pt] (-2.30,3.60) rectangle (-1.15,4.20);
  \fill[framec!60] (-1.15,3.73) -- (-0.72,3.90) -- (-1.15,4.07) -- cycle;
  \node[font=\small, text=textc] at (-1.72,3.90) {\textbf{Torch}};

  \fill[lightc, opacity=0.17]
    (-0.72,3.90) -- (1.0,2.20) -- (1.0,5.60) -- cycle;
  \draw[lightc!80!orange, line width=0.75pt, opacity=0.50]
    (-0.72,3.90) -- (1.0,2.20);
  \draw[lightc!80!orange, line width=0.75pt, opacity=0.50]
    (-0.72,3.90) -- (1.0,5.60);

  \fill[lightc, opacity=0.20] (1.0,3.45) rectangle (7.5,4.45);
  \draw[lightc!70!orange, line width=0.85pt, opacity=0.60]
    (1.0,3.45) -- (7.5,3.45);
  \draw[lightc!70!orange, line width=0.85pt, opacity=0.60]
    (1.0,4.45) -- (7.5,4.45);
  \node[font=\scriptsize, text=textc!75] at (4.25,4.72)
    {illumination\,/\,scan};

  \fill[lightc, opacity=0.10]
    (7.5,3.45) -- (7.5,4.45) -- (14.4,1.50) -- (14.4,6.30) -- cycle;
  \draw[lightc!65!orange, line width=0.65pt, opacity=0.33]
    (7.5,3.45) -- (14.4,1.50);
  \draw[lightc!65!orange, line width=0.65pt, opacity=0.33]
    (7.5,4.45) -- (14.4,6.30);

  \fill[wallc] (10.9,0.45) -- (14.5,1.35) -- (14.5,8.20) -- (10.9,7.30) -- cycle;
  \draw[wallborder, line width=0.75pt]
    (10.9,0.45) -- (14.5,1.35) -- (14.5,8.20) -- (10.9,7.30) -- cycle;
  \node[font=\small, text=textc, align=center] at (12.80,8.55)
    {Projected sorted output};

  %

  \fill[shadowc!75] (11.20,1.20) ellipse (0.24 and 0.09);
  \node[font=\small, anchor=east, text=textc] at (10.80,1.20) {$0$};
  \draw[->, shadowc!50, line width=0.55pt] (10.83,1.20) -- (10.93,1.20);

  \fill[shadowc!75] (11.58,2.16) ellipse (0.24 and 0.09);
  \fill[shadowc!75] (12.10,2.16) ellipse (0.24 and 0.09);
  \node[font=\small, anchor=east, text=textc] at (10.80,2.16) {$1,\!1$};
  \draw[->, shadowc!50, line width=0.55pt] (10.83,2.16) -- (11.30,2.16);

  \fill[shadowc!80] (12.00,4.08) ellipse (0.24 and 0.09);
  \fill[shadowc!80] (12.52,4.08) ellipse (0.24 and 0.09);
  \fill[shadowc!80] (13.04,4.08) ellipse (0.24 and 0.09);
  \node[font=\small, anchor=east, text=textc] at (10.80,4.08) {$3,\!3,\!3$};
  \draw[->, shadowc!50, line width=0.55pt] (10.83,4.08) -- (11.73,4.08);

  \fill[shadowc!75] (12.68,5.04) ellipse (0.24 and 0.09);
  \node[font=\small, anchor=east, text=textc] at (10.80,5.04) {$4$};
  \draw[->, shadowc!50, line width=0.55pt] (10.83,5.04) -- (12.41,5.04);

  \fill[shadowc!75] (13.12,6.00) ellipse (0.24 and 0.09);
  \fill[shadowc!75] (13.64,6.00) ellipse (0.24 and 0.09);
  \node[font=\small, anchor=east, text=textc] at (10.80,6.00) {$5,\!5$};
  \draw[->, shadowc!50, line width=0.55pt] (10.83,6.00) -- (12.85,6.00);

  \draw[->, framec!65, line width=0.80pt] (8.30,3.90) -- (10.75,3.90);
  \node[font=\scriptsize, text=textc!65, above] at (9.52,3.90) {reveals};

\end{tikzpicture}
\caption{\textbf{Abacus-and-torch analogy for \dialsort{}.}
After ingestion, each value $k$ occupies $H[k]$ independent dials,
each holding one bead at the height corresponding to value~$k$:
$H[0]{=}1$, $H[1]{=}2$, $H[2]{=}0$, $H[3]{=}3$, $H[4]{=}1$, $H[5]{=}2$.
A torch sweeps the value axis from $k{=}0$ to $U{-}1$; repeated values
appear side-by-side on their own dials at the same height.
The shadow projected on the wall is the sorted output
$\langle 0,\,1,1,\,3,3,3,\,4,\,5,5 \rangle$ ---
each mark corresponds to exactly one dial.
No magnitude comparison is ever evaluated; order is revealed, not computed.}
\label{fig:abacus}
\end{figure*}

\begin{quote}
\itshape \dialsort{} does not traverse a sorted sequence.
It illuminates one.
\end{quote}


\subsection{Substrate-1: CPU Software (validated)}

$\bigO{n+U}$, zero order comparisons, C++17 (Section~\ref{sec:results}).
\textbf{Embedding status:} $\phi$ is an order-embedding
$\phi : D \hookrightarrow [0,U-1]$ in the general case
(Corollary~\ref{cor:general}).

\begin{simbox}
\small\textbf{Simulator S1\,\&\,S2 --- DialSort Chip Visualizer}\\
\url{https://elmaestrotic.github.io/dsort/simulators/chip-visualizer/}\\[2pt]
Illustrates the complete pipeline: binary encoding $\to$ one-hot
decoder $\to$ activation matrix $\to$ histogram bank $\to$ bottom-up
scanner.  Single-lane mode $\equiv$ Substrate-1 (CPU); multi-lane mode
$\equiv$ Substrate-2 (FPGA parallel lanes).  No magnitude comparison at
any stage.
\end{simbox}

\subsection{Substrate-2: FPGA (analytical projection)\label{sec:s2}}

$k$ parallel ingestion lanes feed the CRN each clock cycle.  On Xilinx
Alveo U50 (XCU50, UltraScale+, 16\,nm, 320\,MHz, 8\,GB HBM2,
$\approx\!460$\,GB/s~\cite{xilinx_u50}): for $U{=}1024$, the CRN with
$k{=}16$ lanes has $\lceil\log_2 16\rceil=4$ pipeline stages
$= 12.5$\,ns at 320\,MHz.  A scan of $U{=}1024$ registers completes in
$\approx\!3.2\,\mu$s.  \textbf{All FPGA figures are analytical
projections; no RTL synthesis has been performed.}

\textbf{Embedding status:} $\phi$ is an order-embedding, same
conditions as Substrate-1.

\subsection{Substrate-3: Photonic (research direction)\label{sec:s3}}

The scan signal is a light pulse propagating through a waveguide.
Histogram positions with nonzero count are optical ring resonators in
resonance; empty positions are transparent.

\textbf{Isomorphism in Substrate-3.}  As established in Corollary~\ref{cor:photonic}, the mathematical
order-isomorphism $\phi : D \to [0,U-1]$ holds whenever $D=[0,U-1]$,
regardless of multiplicities (since duplicates live in $H$, not $D$).
Substrate-3 additionally achieves a \emph{physical} one-to-one
correspondence: when $H[k]\in\{0,1\}$ for all $k$ (binary presence
mode, no repeated keys), each active ring resonator corresponds
bijectively to exactly one element of $D$, and that bijection preserves
$\leq$.  When $H[k]\geq 2$ for some $k$, the mathematical
order-isomorphism on $D$ is unaffected, but the physical one-to-one
correspondence breaks --- a single resonator cannot represent
multiplicity.  Substrate-3 is therefore the unique substrate in which
\emph{both} the mathematical order-isomorphism \emph{and} the physical
one-to-one realisation hold simultaneously.
\textbf{This substrate is a research direction, not a near-term
engineering target.}

\begin{simbox}
\small\textbf{Simulator S3 --- Sovereign Dial Model}\\
\url{https://elmaestrotic.github.io/dsort/simulators/sovereign-dial/}\\[2pt]
Each element falls simultaneously into its sovereign tube at the height
of its value --- no staircase, no inter-tube interaction.  The light
beam represents the photonic wave; active resonators emit by resonance,
no comparator, no control signal.  The binary presence case (each value
at most once) illustrates Corollary~\ref{cor:photonic}: $\varphi$ is a
full order-isomorphism.  The manual photonic beam slider enables
level-by-level exploration.
\end{simbox}

\subsection{Substrate-4: Cerebras WSE-3 (analytical
projection)\label{sec:s4}}

The Cerebras WSE-3 is a wafer-scale processor on TSMC 5\,nm containing
900{,}000 AI-optimised compute tiles arranged in a 2D mesh.  Each tile
has \textbf{48\,KB of local SRAM} accessible in a single clock cycle,
exclusively by that tile.  Aggregate on-chip memory bandwidth exceeds
21\,PB/s; fabric (inter-tile) bandwidth is 214\,Pb/s~\cite{cerebras_wse3,cerebras_sdk}.

\textbf{Architectural correspondence.}
Assign one tile to each universe position $k$: tile $k$ stores the
counter $H[k]$ in its local SRAM (4\,bytes; the remaining
$\approx\!48$\,KB is available for CRN pipeline buffering).  An
incoming key $k$ arrives as a \emph{wavelet} --- a 32-bit message in
WSE-3 terminology --- routed directly to tile $k$.  The tile increments
its counter in one clock cycle.  No inter-tile communication is
required for ingestion; tiles are fully independent.  The CRN topology
maps naturally onto the 2D mesh fabric: equal-key wavelets from
adjacent entries can be merged by the fabric routing logic,
realising Definition~\ref{sec:crn_def} in hardware.

\begin{table}[htbp]
\centering
\caption{WSE-3 analytical projection for \dialsort{} (Substrate-4).
All values derived from published specifications~\cite{cerebras_wse3,cerebras_sdk}.
\textbf{No physical experiment has been conducted.}}
\label{tab:wse3}
\small
\begin{tabular}{@{}p{2.6cm}p{1.5cm}p{2.6cm}@{}}
\toprule
\textbf{Parameter} & \textbf{Value} & \textbf{Derivation} \\
\midrule
Tiles ($U{=}65536$) & 65,536 & 7.3\% of 900k \\
SRAM for $H[k]$    & 4\,B   & int32; 48\,KB avail. \\
CRN stages ($k{=}32$) & 5 cycles & $\lceil\log_2 32\rceil$ \\
Ingest ($n{=}10^7$, $k{=}32$) & $\approx\!1.04$\,ms & $n/k$ @ 300\,MHz \\
Scan ($U{=}65536$) & $\approx\!218\,\mu$s & 65,536 cycles \\
Upper-bound ingress & $\sim\!9{,}600$\,M/s & $32\times300$\,MHz \\
\multicolumn{1}{l}{\footnotesize\textit{throughput (idealized)}} & & \\
PCIe bottleneck    & $\approx\!2.5$\,ms & 40\,MB @ 16\,GB/s \\
\bottomrule
\end{tabular}
\end{table}

\begin{proposition}[Fixed-topology analytical projection for
WSE-3\label{thm:wse3}]
The following result is an analytical hardware projection, not an
asymptotic theorem in the RAM model.

Assume a WSE-3 deployment in which:
\begin{enumerate}
  \item each universe position $k \in [0,U-1]$ is assigned to a
        dedicated tile;
  \item tile $k$ stores $H[k]$ in local SRAM;
  \item incoming keys are routed as wavelets to their corresponding
        destination tiles;
  \item routing latency is bounded above by a hardware-dependent
        constant $C_{\mathrm{route}}$ determined by the fixed mesh
        topology and per-hop communication latency; in a 2D mesh of
        $T$ tiles the diameter is $O(\sqrt{T})$, so for $T = 900{,}000$
        fixed tiles $C_{\mathrm{route}}$ is bounded by at most
        $\approx\!949$ hops --- a hardware constant independent of the
        number of input keys $n$;
  \item host-device transfer is external to on-chip ingestion.
\end{enumerate}

Then, for fixed WSE-3 hardware and $U \leq 900{,}000$, the ingestion
latency is bounded by
\[
T_{\mathrm{ingest}}^{\mathrm{WSE3}} \leq C_{\mathrm{route}}
+ C_{\mathrm{update}},
\]
where both constants depend on the fixed hardware topology and routing
implementation, \emph{not} on asymptotic $n \to \infty$ growth.
The end-to-end wall-clock time is therefore
\[
T_{\mathrm{total}}^{\mathrm{WSE3}}
= T_{\mathrm{transfer}}
+ T_{\mathrm{ingest}}^{\mathrm{WSE3}}
+ T_{\mathrm{scan}}(U),
\]
with $T_{\mathrm{scan}}(U) = O(U)$ for a serial ordered scan, and
$T_{\mathrm{transfer}}$ the host-device PCIe transfer term
($\approx\!2.5$\,ms for 40\,MB at $N{=}10^7$, which dominates for
large inputs).
\end{proposition}

\begin{proof}
WSE-3 is a fixed wafer-scale architecture with a finite number of tiles
arranged in a fixed communication fabric.  Once one tile is assigned
per universe position, each key is routed to a predetermined
destination tile and incremented locally in $O(1)$ SRAM cycles.
Because the hardware topology is fixed and finite, the routing latency
has a hardware-determined upper bound $C_{\mathrm{route}}$ independent
of $n$.  This does not justify an asymptotic $O(1)$-in-$n$ claim in
the RAM-model sense; it justifies a bounded-latency hardware projection
under fixed-topology assumptions.  The variable costs visible at the
system level are $T_{\mathrm{scan}}(U) = O(U)$ and the host-device
transfer $T_{\mathrm{transfer}}$.
\end{proof}

\begin{remark}
This proposition is an analytical projection from published WSE-3
specifications~\cite{cerebras_wse3,cerebras_sdk}; no physical
experiment has been conducted.  The ingestion term does not grow with
$n$ in the fixed-hardware model, but this should not be interpreted as
algorithmic $O(1)$-in-$n$ in the standard sense: it reflects a
hardware constant, not a vanishing asymptotic.
\end{remark}

\textbf{Embedding status.}  The canonical embedding $\phi : D \hookrightarrow
[0,U-1]$ holds universally (Proposition~\ref{prop:iso}).  The full
order-isomorphism (Corollary~\ref{cor:photonic}) does not hold in
Substrate-4 when $D \subsetneq [0,U-1]$ or when keys repeat, since
multiple wavelets map to the same tile counter.

\begin{simbox}
\small\textbf{Simulator S4 --- DialSort Luminous (WSE-3 model)}\\
\url{https://elmaestrotic.github.io/dsort/simulators/luminico/}\\[2pt]
Each dial (vertical rail) represents one WSE-3 tile.  Luminous
intensity is proportional to $H[k]$, directly analogous to the tile's
SRAM counter value.  Duplicate keys accumulate brightness in a single
tile, corresponding to the additive counter increment.  The bottom-up
scanner reveals sorted order.  The manual pulse slider inspects
individual tile states at any scan level.
\end{simbox}

\section{Experimental Results}
\label{sec:results}

\subsection{Methodology}

\textbf{Scope.}  All results are Substrate-1 (CPU software).  No FPGA
hardware or RTL simulation was used.

\noindent\textbf{Hardware.}  Intel x86-64, 8 hardware threads,
\texttt{g++ 11.4.0} (\texttt{-O3 -std=c++17 -pthread}).  No SIMD flags.

\noindent\textbf{Timing.}  \texttt{std::chrono::high\_resolution\_clock}.
3 warmup rounds discarded; best-of-7 reported in the main comparison
tables.  Mean and standard deviation over all 7 runs are computed and
reported separately in Section~\ref{sec:stability} to assess
measurement stability.
\texttt{check\_sorted()} called after every run.
Seed fixed at \texttt{20260321}; full benchmark v3.1 released.

\noindent\textbf{Baselines.}  \texttt{std::sort} (GCC introsort),
IPS$^4$o (sequential, header-only, benchmarked directly in
Section~\ref{sec:vs_ips4o}), and ska\_sort (benchmarked directly in
Section~\ref{sec:vs_ska}).

\subsection{Distributions}

Four distributions over universe sizes $U \in \{256, 1024, 65536\}$ and
input sizes $N \in \{10^4, 10^5, 10^6, 10^7\}$:
\textit{uniform}, \textit{skewed} (80\% of keys drawn uniformly at
random from $[0,\lfloor 0.05 \cdot U \rfloor]$ and the remaining
20\% uniformly from $[0,U{-}1]$ --- a controlled mixture, not a named
parametric distribution), \textit{sorted}, and \textit{reverse-sorted}.

\subsection{Bounded-universe results}

Table~\ref{tab:selected} shows selected peak results.
\dialsort{}-Parallel achieves its global peak speedup of
\textbf{39.77$\times$} at $N{=}10^7$, $U{=}1024$, uniform.
Under skewed distributions it reaches \textbf{37.97$\times$} with peak
throughput of \textbf{115.9~M~keys/s}.  \dialsort{} peaks at
\textbf{30.26$\times$}.

Across all 48 \dialsort{} configurations, it outperforms
\texttt{std::sort} in 44 of 48 cases.  The 4 exceptions occur at
$N{=}10^4$, $U{=}65536$, where the $O(U)$ scan cost dominates a
small input --- the algorithm's expected domain boundary.
Note: this comparison is against \texttt{std::sort}; the IPS$^4$o
comparison is presented separately in Section~\ref{sec:vs_ips4o}.

\begin{table}[htbp]
\centering
\caption{Selected peak results --- bounded-universe sorting.
$\star$~=~global peak.  Full 192-row dataset in supplementary CSV.
All correctness checks: \textsc{passed}.}
\label{tab:selected}
\small
\begin{tabular}{@{}llrrr@{}}
\toprule
\textbf{Algorithm} & \textbf{Dist.} &
\textbf{$N$} & \textbf{M/s} & \textbf{Speedup} \\
\midrule
DS-Parallel  & skewed  & $10^7$ & 115.9 & 37.97$\times$ \\
DS-Parallel  & uniform & $10^7$ & 115.1 & \textbf{39.77$\times$ $\star$}\\
DialSort     & uniform & $10^7$ & 87.6  & 30.26$\times$ \\
DS-Parallel  & skewed  & $10^6$ & 112.5 & 29.90$\times$ \\
DialSort     & skewed  & $10^7$ & 87.7  & 28.71$\times$ \\
\texttt{std::sort} & uniform & $10^7$ & 2.9 & 1.00$\times$ \\
\bottomrule
\end{tabular}
\end{table}

\textbf{Scaling analysis.}  Speedup grows monotonically with $N$ for
fixed $U$, as expected from $O(n{+}U)$ vs $O(n \log n)$.  The parallel
variant (8 threads) outperforms the sequential variant in all
$N \geq 10^6$ configurations; the crossover occurs around
$N{\approx}10^5$ where thread-management overhead amortizes.

\textbf{Cache behavior.}  For $U{=}256$, $H$ occupies 1~KB entirely
in L1 cache.  For $U{=}1024$, $H$ is 4~KB, still L1-resident.
For $U{=}65536$, $H$ is 256~KB, overflowing to L2; speedup reduction
is consistent with increased cache miss rate.

\subsection{Full int32 range}

\dialsort{}-Radix (LSD, base~256, 4 passes, zero order comparisons)
handles the full signed int32 range.  Speedup grows from 9.52$\times$
at $N{=}10^4$ to \textbf{14.61$\times$} at $N{=}10^7$.

\begin{table}[htbp]
\centering
\caption{\dialsort{}-Radix vs \texttt{std::sort}, full int32 range.
All correctness checks: \textsc{passed}.}
\label{tab:radix}
\small
\begin{tabular}{@{}lrrrr@{}}
\toprule
\textbf{Algorithm} & \textbf{$N$} &
\textbf{ms} & \textbf{M/s} & \textbf{Speedup} \\
\midrule
DS-Radix     & $10^4$ & 0.250   & 40.1 & 9.52$\times$ \\
\texttt{std::sort} & $10^4$ & 2.376 & 4.2 & 1.00$\times$ \\
\addlinespace
DS-Radix     & $10^5$ & 2.884   & 34.7 & 10.62$\times$ \\
\texttt{std::sort} & $10^5$ & 30.613 & 3.3 & 1.00$\times$ \\
\addlinespace
DS-Radix     & $10^6$ & 27.262  & 36.7 & 13.18$\times$ \\
\texttt{std::sort} & $10^6$ & 359.409 & 2.8 & 1.00$\times$ \\
\addlinespace
DS-Radix     & $10^7$ & 261.983 & 38.2 & \textbf{14.61$\times$} \\
\texttt{std::sort} & $10^7$ & 3827.103 & 2.6 & 1.00$\times$ \\
\bottomrule
\end{tabular}
\end{table}

\section{\dialsort{} vs.\ Classic Counting Sort: Direct Comparison}
\label{sec:vs_classic}

\subsection{Motivation}

A common reviewer objection is that \dialsort{} is structurally
equivalent to Classic Counting Sort (CLRS \S8.2).
Section~\ref{sec:related} addresses this formally via Table~\ref{tab:related}.
This section provides the empirical counterpart: a direct,
side-by-side benchmark of both algorithms across all 48 configurations
($N \in \{10^4, 10^5, 10^6, 10^7\}$,
$U \in \{256, 1024, 65536\}$,
four distributions), using identical seeds and timing methodology.

\subsection{Structural cost model}

Table~\ref{tab:struct_cost} summarizes the algorithmic differences that
the benchmark is designed to quantify.

\begin{table}[htbp]
\centering
\caption{Structural cost comparison: \dialsort{}
vs.\ Classic Counting Sort (CLRS \S8.2).
These differences are independent of implementation quality;
they are inherent to each algorithm's design.}
\label{tab:struct_cost}
\small
\begin{tabular}{@{}lcc@{}}
\toprule
\textbf{Property} &
\textbf{\dialsort{}} &
\textbf{Classic CS} \\
\midrule
Passes over data     & \textbf{2}            & 3 \\
Prefix-sum pass      & \textbf{None}         & Mandatory \\
Output allocation    & \textbf{In-place}     & $O(n)$ array \\
Auxiliary memory     & $O(U)$                & $O(U+n)$ \\
Stability            & Not stable            & Stable \\
Histogram role       & \textbf{Canonical repr.} & Intermediate \\
\bottomrule
\end{tabular}
\end{table}

\noindent
The key differences are: (1) \dialsort{} eliminates the mandatory
prefix-sum pass entirely; (2) it writes the output in-place into the
input array, avoiding the $O(n)$ scatter buffer and copy-back that
Classic CS requires; (3) its total memory footprint is $O(U)$
vs.\ $O(U+n)$.

\subsection{Benchmark results}

\begin{table}[htbp]
\centering
\caption{Selected results --- \dialsort{} vs.\ Classic Counting
Sort vs.\ \texttt{std::sort}.  Ratio~$=$ \dialsort{} ms $\div$ Classic ms
($<1.0$ means \dialsort{} is faster).  Full 48-configuration dataset in
supplementary CSV.  All correctness checks: \textsc{passed}.}
\label{tab:vs_classic}
\small
\setlength{\tabcolsep}{4pt}
\begin{tabular}{@{}llrrrrr@{}}
\toprule
\textbf{Algo} & \textbf{Dist} & \textbf{$N$} & \textbf{$U$} &
\textbf{ms} & \textbf{Spd} & \textbf{Ratio} \\
\midrule
DialSort  & uniform & $10^7$ & 1024 & 144.2 & 26.2$\times$ & \textbf{0.42} \\
Classic   & uniform & $10^7$ & 1024 & 343.3 & 11.0$\times$ & --- \\
\texttt{std::sort} & uniform & $10^7$ & 1024 & 3779 & 1.0$\times$ & --- \\
\addlinespace
DialSort  & skewed  & $10^7$ & 65536 & 149.0 & 29.0$\times$ & \textbf{0.44} \\
Classic   & skewed  & $10^7$ & 65536 & 339.7 & 12.7$\times$ & --- \\
\texttt{std::sort} & skewed  & $10^7$ & 65536 & 4324 & 1.0$\times$ & --- \\
\addlinespace
DialSort  & sorted  & $10^6$ & 1024 & 13.9 & 15.3$\times$ & \textbf{0.23} \\
Classic   & sorted  & $10^6$ & 1024 & 60.0 & 3.5$\times$ & --- \\
\texttt{std::sort} & sorted  & $10^6$ & 1024 & 212.9 & 1.0$\times$ & --- \\
\addlinespace
DialSort  & reverse & $10^5$ & 65536 & 4.69 & 3.9$\times$ & 1.06 \\
Classic   & reverse & $10^5$ & 65536 & 4.44 & 4.1$\times$ & --- \\
\texttt{std::sort} & reverse & $10^5$ & 65536 & 18.3 & 1.0$\times$ & --- \\
\bottomrule
\end{tabular}
\end{table}

Across all 48 configurations, \dialsort{} outperforms Classic
Counting Sort in \textbf{46 of 48} cases.  The single case where
Classic CS wins (ratio $= 1.058$, $N{=}10^5$, $U{=}65536$,
reverse-sorted) is marginal and falls precisely at the boundary of
\dialsort{}'s effective domain ($N \approx U$), where the $O(U)$ scan
cost is comparable to the $O(n)$ scatter cost.  The overall summary is:

\begin{itemize}
\item \textbf{Min ratio}: 0.232 ($N{=}10^6$, $U{=}1024$, sorted) ---
      \dialsort{} more than $4\times$ faster than Classic CS.
\item \textbf{Average ratio}: 0.605 --- \dialsort{} $\sim$39\% faster
      on average across all configurations.
\item \textbf{46/48} configurations: \dialsort{} strictly faster.
\end{itemize}

\subsection{Performance visualizations}

Figure~\ref{fig:ratio_by_N} shows the average D/C ratio (\dialsort{} ms
$\div$ Classic CS ms) as a function of $N$.  The ratio decreases as $N$
grows, confirming that \dialsort{}'s structural savings scale
favorably with input size.

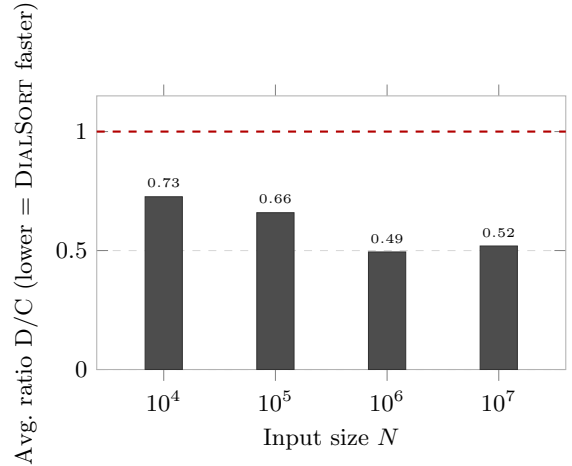
\begin{figure}[htbp]
\centering
\begin{tikzpicture}
\begin{axis}[
    width=\columnwidth, height=5.2cm,
    ybar, bar width=14pt,
    symbolic x coords={$10^4$,$10^5$,$10^6$,$10^7$},
    xtick=data,
    ymin=0, ymax=1.15,
    ylabel={Avg.\ ratio D/C (lower = \dialsort{} faster)},
    xlabel={Input size $N$},
    ylabel style={font=\small},
    xlabel style={font=\small},
    tick label style={font=\small},
    nodes near coords,
    nodes near coords style={font=\tiny, /pgf/number format/fixed,
                             /pgf/number format/precision=2},
    ymajorgrids=true, grid style={dashed,gray!30},
    axis line style={gray!60},
    enlarge x limits=0.2,
]
\addplot[fill=black!70, draw=black!80]
    coordinates {($10^4$,0.726) ($10^5$,0.659) ($10^6$,0.494) ($10^7$,0.519)};
\draw[dashed, red!70!black, thick] ({rel axis cs:0,0} |- {axis cs:$10^4$,1.0}) --
    ({rel axis cs:1,0} |- {axis cs:$10^7$,1.0})
    node[right, font=\tiny, red!70!black] {parity};
\end{axis}
\end{tikzpicture}
\caption{Average ratio \dialsort{} ms $\div$ Classic CS ms,
by input size $N$.
Values below the dashed parity line indicate \dialsort{} is faster.
The ratio decreases with $N$, confirming that the savings from
eliminating the prefix-sum pass and the $O(n)$ output buffer grow with
input size.}
\label{fig:ratio_by_N}
\end{figure}

Figure~\ref{fig:speedup_grouped} shows the speedup of both algorithms
over \texttt{std::sort} at $N{=}10^7$.  \dialsort{}
consistently achieves higher speedup than Classic CS across all
configurations at the largest input size.

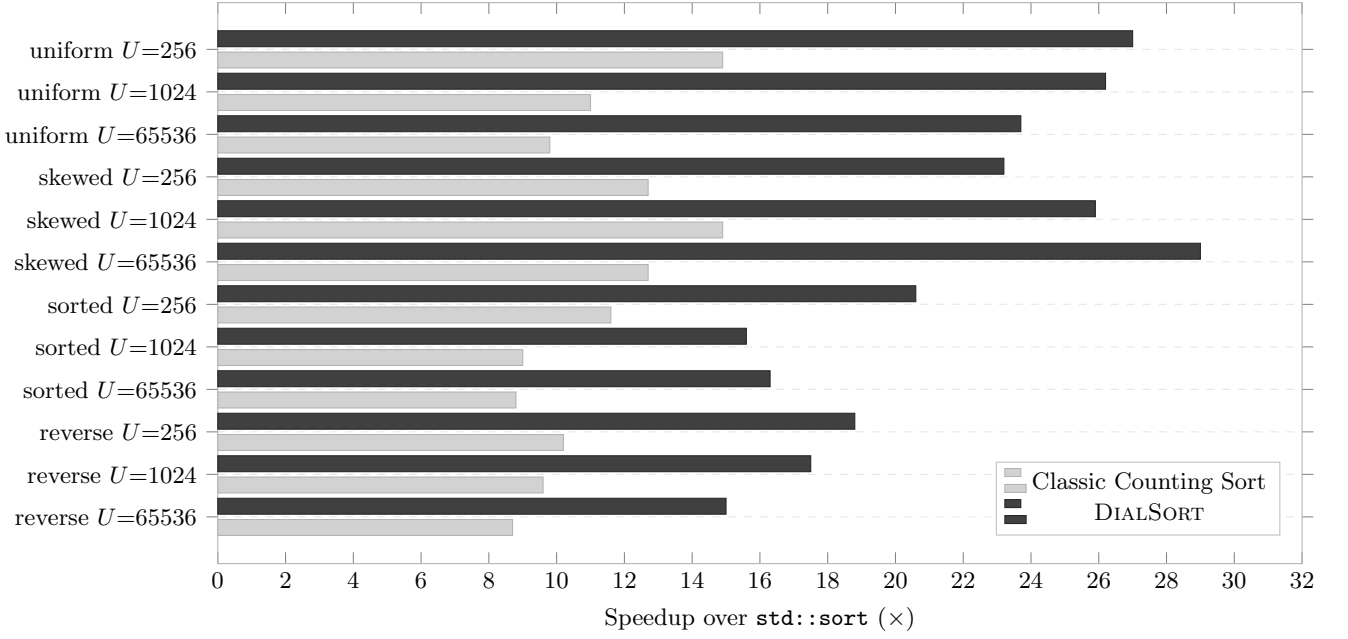
\begin{figure*}[tp]
\centering
\begin{tikzpicture}
\begin{axis}[
    width=\textwidth, height=9cm,
    xbar, bar width=6pt,
    ytick=data,
    yticklabels={
        uniform $U{=}256$, uniform $U{=}1024$, uniform $U{=}65536$,
        skewed $U{=}256$, skewed $U{=}1024$, skewed $U{=}65536$,
        sorted $U{=}256$, sorted $U{=}1024$, sorted $U{=}65536$,
        reverse $U{=}256$, reverse $U{=}1024$, reverse $U{=}65536$
    },
    yticklabel style={font=\small},
    xticklabel style={font=\small},
    xlabel={Speedup over \texttt{std::sort} ($\times$)},
    xlabel style={font=\small},
    xmin=0, xmax=32,
    legend style={at={(0.98,0.05)}, anchor=south east, font=\small,
                  draw=gray!50, fill=white},
    ymajorgrids=true, grid style={dashed,gray!20},
    axis line style={gray!60},
    y dir=reverse,
]
\addplot[fill=gray!35, draw=gray!60, bar shift=-4pt]
    coordinates {
        (14.9,0) (11.0,1) (9.8,2)
        (12.7,3) (14.9,4) (12.7,5)
        (11.6,6) (9.0,7)  (8.8,8)
        (10.2,9) (9.6,10) (8.7,11)
    };
\addplot[fill=black!75, draw=black!85, bar shift=4pt]
    coordinates {
        (27.0,0) (26.2,1) (23.7,2)
        (23.2,3) (25.9,4) (29.0,5)
        (20.6,6) (15.6,7) (16.3,8)
        (18.8,9) (17.5,10)(15.0,11)
    };
\legend{Classic Counting Sort, \dialsort{}}
\end{axis}
\end{tikzpicture}
\caption{Speedup over \texttt{std::sort} at $N{=}10^7$:
\dialsort{} (dark) vs.\ Classic Counting Sort (light),
grouped by distribution and universe size.
\dialsort{} achieves higher speedup in all 12 configurations.
The largest gap appears at large $U$ (65536) under skewed
distributions.}
\label{fig:speedup_grouped}
\end{figure*}

Figure~\ref{fig:ratio_by_U} decomposes the average D/C ratio by
universe size $U$.

\begin{figure}[htbp]
\centering
\begin{tikzpicture}
\begin{axis}[
    width=\columnwidth, height=5.2cm,
    ybar, bar width=18pt,
    symbolic x coords={$U{=}256$,$U{=}1024$,$U{=}65536$},
    xtick=data,
    ymin=0, ymax=1.1,
    ylabel={Avg.\ ratio D/C},
    xlabel={Universe size $U$},
    ylabel style={font=\small},
    xlabel style={font=\small},
    tick label style={font=\small},
    nodes near coords,
    nodes near coords style={font=\tiny, /pgf/number format/fixed,
                             /pgf/number format/precision=2},
    ymajorgrids=true, grid style={dashed,gray!30},
    axis line style={gray!60},
    enlarge x limits=0.3,
    legend style={at={(0.97,0.97)}, anchor=north east, font=\small},
]
\addplot[fill=black!65, draw=black!80]
    coordinates {($U{=}256$,0.549) ($U{=}1024$,0.546) ($U{=}65536$,0.742)};
\draw[dashed, red!70!black, thick]
    ({rel axis cs:0,0} |- {axis cs:$U{=}256$,1.0}) --
    ({rel axis cs:1,0} |- {axis cs:$U{=}65536$,1.0})
    node[right, font=\tiny, red!70!black] {parity};
\end{axis}
\end{tikzpicture}
\caption{Average ratio D/C by universe size $U$.
\dialsort{}'s advantage is largest at $U{=}256$
(L1-resident histogram) and decreases as $U$ grows.
All ratios remain below parity.}
\label{fig:ratio_by_U}
\end{figure}
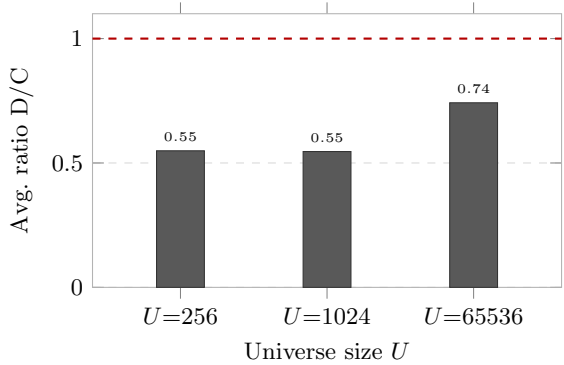

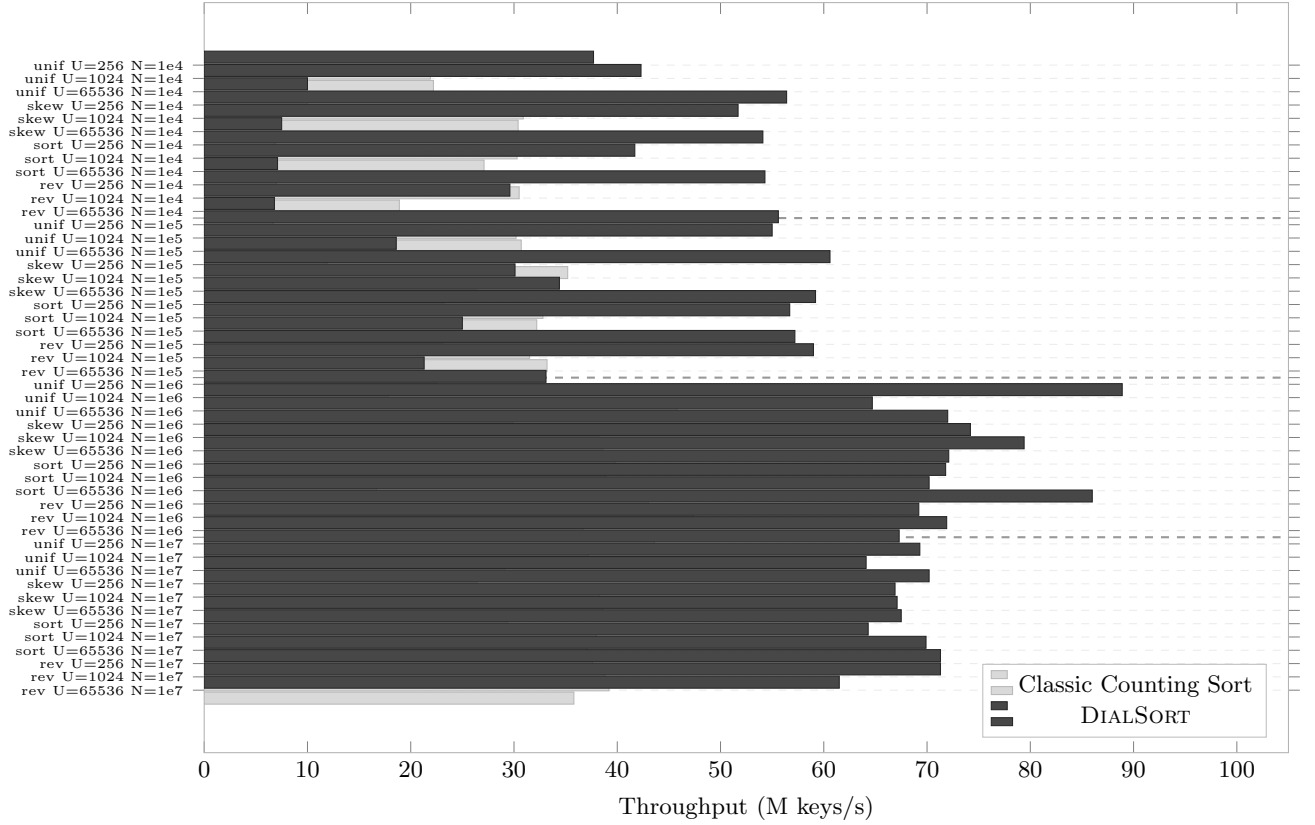
\begin{figure*}[tp]
\centering
\begin{tikzpicture}
\begin{axis}[
    width=\textwidth, height=11.5cm,
    xbar, bar width=4.5pt,
    ytick=data,
    yticklabels={
        unif U=256 N=1e4,  unif U=1024 N=1e4,  unif U=65536 N=1e4,
        skew U=256 N=1e4,  skew U=1024 N=1e4,  skew U=65536 N=1e4,
        sort U=256 N=1e4,  sort U=1024 N=1e4,  sort U=65536 N=1e4,
        rev  U=256 N=1e4,  rev  U=1024 N=1e4,  rev  U=65536 N=1e4,
        unif U=256 N=1e5,  unif U=1024 N=1e5,  unif U=65536 N=1e5,
        skew U=256 N=1e5,  skew U=1024 N=1e5,  skew U=65536 N=1e5,
        sort U=256 N=1e5,  sort U=1024 N=1e5,  sort U=65536 N=1e5,
        rev  U=256 N=1e5,  rev  U=1024 N=1e5,  rev  U=65536 N=1e5,
        unif U=256 N=1e6,  unif U=1024 N=1e6,  unif U=65536 N=1e6,
        skew U=256 N=1e6,  skew U=1024 N=1e6,  skew U=65536 N=1e6,
        sort U=256 N=1e6,  sort U=1024 N=1e6,  sort U=65536 N=1e6,
        rev  U=256 N=1e6,  rev  U=1024 N=1e6,  rev  U=65536 N=1e6,
        unif U=256 N=1e7,  unif U=1024 N=1e7,  unif U=65536 N=1e7,
        skew U=256 N=1e7,  skew U=1024 N=1e7,  skew U=65536 N=1e7,
        sort U=256 N=1e7,  sort U=1024 N=1e7,  sort U=65536 N=1e7,
        rev  U=256 N=1e7,  rev  U=1024 N=1e7,  rev  U=65536 N=1e7
    },
    yticklabel style={font=\tiny},
    xticklabel style={font=\small},
    xlabel={Throughput (M~keys/s)},
    xlabel style={font=\small},
    xmin=0, xmax=105,
    legend style={at={(0.98,0.02)}, anchor=south east, font=\small,
                  draw=gray!50, fill=white},
    ymajorgrids=true, grid style={dashed,gray!15},
    axis line style={gray!60},
    y dir=reverse,
    extra y ticks={11.5, 23.5, 35.5},
    extra y tick labels={},
    extra y tick style={grid=major, grid style={black!40, thick}},
]
\addplot[fill=gray!30, draw=gray!55, bar shift=-3pt] coordinates {
    (21.9,0)  (22.2,1)  (10.0,2)
    (30.9,3)  (30.4,4)  (6.9,5)
    (30.3,6)  (27.1,7)  (7.0,8)
    (30.5,9)  (18.9,10) (6.7,11)
    (30.2,12) (30.7,13) (11.9,14)
    (35.2,15) (26.4,16) (23.2,17)
    (32.8,18) (32.2,19) (23.1,20)
    (31.5,21) (33.2,22) (22.5,23)
    (17.8,24) (45.8,25) (29.9,26)
    (38.2,27) (38.6,28) (39.8,29)
    (38.9,30) (16.7,31) (43.0,32)
    (47.4,33) (36.7,34) (43.6,35)
    (37.2,36) (29.1,37) (26.4,38)
    (38.3,39) (38.6,40) (29.4,41)
    (38.0,42) (37.1,43) (37.6,44)
    (38.8,45) (39.2,46) (35.8,47)
};
\addplot[fill=black!72, draw=black!85, bar shift=3pt] coordinates {
    (37.7,0)  (42.3,1)  (10.0,2)
    (56.4,3)  (51.7,4)  (7.5,5)
    (54.1,6)  (41.7,7)  (7.1,8)
    (54.3,9)  (29.6,10) (6.8,11)
    (55.6,12) (55.0,13) (18.6,14)
    (60.6,15) (30.1,16) (34.4,17)
    (59.2,18) (56.7,19) (25.0,20)
    (57.2,21) (59.0,22) (21.3,23)
    (33.1,24) (88.9,25) (64.7,26)
    (72.0,27) (74.2,28) (79.4,29)
    (72.1,30) (71.8,31) (70.2,32)
    (86.0,33) (69.2,34) (71.9,35)
    (67.3,36) (69.3,37) (64.1,38)
    (70.2,39) (66.9,40) (67.1,41)
    (67.5,42) (64.3,43) (69.9,44)
    (71.3,45) (71.3,46) (61.5,47)
};
\legend{Classic Counting Sort, \dialsort{}}
\end{axis}
\end{tikzpicture}
\caption{
Head-to-head throughput comparison: \dialsort{} (dark) vs.\
Classic Counting Sort (light) across all 48 benchmark configurations.
\dialsort{} wins in 46 of 48 cases.
}
\label{fig:head_to_head}
\end{figure*}

\subsection{Interpretation}

The empirical results confirm the structural cost model
(Table~\ref{tab:struct_cost}).  The two main sources of \dialsort{}'s
advantage over Classic CS are:

\begin{enumerate}
\item \textbf{Eliminated prefix-sum pass.}  Classic CS requires a
complete sweep of the $O(U)$ histogram to convert counts into output
positions before any element can be placed.  \dialsort{} proceeds
directly to the projection sweep, saving one full $O(U)$ pass.

\item \textbf{In-place output.}  Classic CS scatters elements into a
separate $O(n)$ output array and then copies it back.  \dialsort{}
writes sorted output directly into the input array in a single
forward scan, eliminating both the allocation and the copy.
\end{enumerate}

The single case where Classic CS wins ($N{=}10^5$, $U{=}65536$,
reverse-sorted, ratio $= 1.058$) is explained by the $O(U)$ scan
dominating at that $N/U$ ratio: when $n \approx U$, the extra scan
pass that \dialsort{} saves is offset by Classic CS's simpler
scatter-write pattern for small $n$.  This is consistent with the
theoretical boundary $n \lesssim U$ discussed in
Section~\ref{sec:limits}.

\section{Performance against State-of-the-Art: IPS$^4$o}
\label{sec:vs_ips4o}

\subsection{Motivation}

To position \dialsort{} within the current landscape of high-performance
sorting, we conducted a direct comparison against IPS$^4$o~\cite{ips4o2022},
widely regarded as one of the fastest comparison-based sorters in
practice.  This experiment resolves the limitation identified in
Section~\ref{sec:limits} of the prior version of this paper, where
comparison against IPS$^4$o was listed as future work.

\subsection{Experimental setup}

The benchmark uses IPS$^4$o (sequential, header-only, cloned from the
official repository) compiled under identical conditions
(\texttt{g++ 11.4.0 -O3 -std=c++17 -pthread}) with TBB parallel
support disabled to isolate sequential algorithmic performance.
All 48 configurations ($N \in \{10^4,10^5,10^6,10^7\}$,
$U \in \{256,1024,65536\}$, four distributions) use the same fixed
seed, warmup policy, and best-of-7 timing as the rest of the benchmark
suite.  Full source released alongside this paper.

\subsection{Results}

Table~\ref{tab:vs_ips4o} reports selected configurations at $N{=}10^7$.
Two regimes are evaluated for a fair assessment:

\begin{itemize}
\item \textbf{Sequential regime:} \dialsort{} vs.\ IPS$^4$o-seq
      (both single-threaded).
\item \textbf{Parallel regime:} \dialsort{}-Parallel vs.\
      IPS$^4$o-par (both 8 TBB threads), reported in
      Section~\ref{sec:vs_ips4o_par}.
\end{itemize}

Across all 48 configurations in the sequential regime, \dialsort{}
outperforms IPS$^4$o-seq in \textbf{24 of 48} cases.
Performance is strongly regime-dependent.  On uniform and skewed
distributions at $N \geq 10^6$, \dialsort{} wins in all configurations,
with speedups ranging from $2.9\times$ to $9.1\times$.  On sorted and
reverse-sorted inputs, IPS$^4$o's adaptive introsort exploits
pre-existing order via branch prediction, reaching throughputs
exceeding 1{,}500~M~keys/s and outperforming \dialsort{} in those
configurations.  The sequential losses occur exclusively under sorted,
reverse-sorted inputs, or the small-$N$ regime ($N{=}10^4$, $U{=}65536$)
where the $O(U)$ scan cost dominates.

\begin{table}[htbp]
\centering
\caption{Performance summary at $N{=}10^7$, best-of-7 runs.
Speedup columns are relative to \texttt{std::sort}.
Column ``vs IPS$^4$o'' reports \dialsort{} M/s $\div$ IPS$^4$o M/s
($>1.0$ means \dialsort{} is faster).
All correctness checks: \textsc{passed}.}
\label{tab:vs_ips4o}
\small
\setlength{\tabcolsep}{4pt}
\begin{tabular}{@{}llrrrr@{}}
\toprule
\textbf{Algorithm} & \textbf{Dist} & \textbf{$U$} &
\textbf{M/s} & \textbf{vs std} & \textbf{vs IPS$^4$o} \\
\midrule
\dialsort{}      & uniform & 256   & 521.2 & 25.4$\times$ & \textbf{4.67$\times$} \\
DS-Parallel      & uniform & 256   & 300.4 & 14.7$\times$ & 2.69$\times$ \\
IPS$^4$o-seq     & uniform & 256   & 111.6 & 5.4$\times$  & 1.00$\times$ \\
\addlinespace
\dialsort{}      & skewed  & 256   & 471.2 & 19.2$\times$ & \textbf{3.06$\times$} \\
DS-Parallel      & skewed  & 256   & 390.8 & 15.9$\times$ & 2.54$\times$ \\
IPS$^4$o-seq     & skewed  & 256   & 153.8 & 6.3$\times$  & 1.00$\times$ \\
\addlinespace
\dialsort{}      & uniform & 1024  & 487.3 & 25.3$\times$ & \textbf{5.36$\times$} \\
IPS$^4$o-seq     & uniform & 1024  & 90.9  & 4.7$\times$  & 1.00$\times$ \\
\addlinespace
\dialsort{}      & uniform & 65536 & 413.3 & 29.5$\times$ & \textbf{9.10$\times$} \\
IPS$^4$o-seq     & uniform & 65536 & 45.4  & 3.2$\times$  & 1.00$\times$ \\
\addlinespace
\dialsort{}      & sorted  & 1024  & 324.3 & 4.8$\times$  & 0.20$\times$ \\
IPS$^4$o-seq     & sorted  & 1024  & 1655  & 24.5$\times$ & 1.00$\times$ \\
\bottomrule
\end{tabular}
\end{table}

Figure~\ref{fig:vs_ips4o_bar} plots throughput for all three algorithms
at $N{=}10^7$ across all distribution and universe-size combinations,
making the domain boundary between \dialsort{} and IPS$^4$o visually
apparent.

\begin{figure*}[tp]
\centering
\begin{tikzpicture}
\begin{axis}[
    width=\textwidth, height=9.5cm,
    xbar, bar width=5pt,
    ytick=data,
    yticklabels={
        uniform $U{=}256$,  uniform $U{=}1024$,  uniform $U{=}65536$,
        skewed $U{=}256$,   skewed $U{=}1024$,   skewed $U{=}65536$,
        sorted $U{=}256$,   sorted $U{=}1024$,   sorted $U{=}65536$,
        reverse $U{=}256$,  reverse $U{=}1024$,  reverse $U{=}65536$
    },
    yticklabel style={font=\small},
    xticklabel style={font=\small},
    xlabel={Throughput (M~keys/s), $N{=}10^7$},
    xlabel style={font=\small},
    xmin=0, xmax=1800,
    legend style={at={(0.98,0.05)}, anchor=south east, font=\small,
                  draw=gray!50, fill=white},
    ymajorgrids=true, grid style={dashed,gray!20},
    axis line style={gray!60},
    y dir=reverse,
]
\addplot[fill=gray!20, draw=gray!45, bar shift=-7pt] coordinates {
    (20.5,0)  (19.3,1)  (14.0,2)
    (24.6,3)  (22.9,4)  (16.1,5)
    (76.0,6)  (67.6,7)  (77.9,8)
    (75.4,9)  (76.6,10) (74.4,11)
};
\addplot[fill=gray!55, draw=gray!70, bar shift=0pt] coordinates {
    (111.6,0)  (90.9,1)   (45.4,2)
    (153.8,3)  (139.5,4)  (68.8,5)
    (1549.0,6) (1655.5,7) (1587.4,8)
    (1007.0,9) (967.8,10) (1025.9,11)
};
\addplot[fill=black!75, draw=black!85, bar shift=7pt] coordinates {
    (521.2,0)  (487.3,1)  (413.3,2)
    (471.2,3)  (524.0,4)  (446.4,5)
    (288.8,6)  (324.3,7)  (357.0,8)
    (352.0,9)  (323.9,10) (345.7,11)
};
\legend{\texttt{std::sort}, IPS$^4$o-seq, \dialsort{}}
\end{axis}
\end{tikzpicture}
\caption{Throughput at $N{=}10^7$: \dialsort{} (dark), IPS$^4$o-seq (medium),
and \texttt{std::sort} (light), across all distribution and universe-size
combinations.
On uniform and skewed inputs \dialsort{} dominates at every $U$.
On sorted/reverse inputs IPS$^4$o exploits pre-existing order via
branch prediction, reaching $>$1500~M~keys/s --- outside \dialsort{}'s
target domain.}
\label{fig:vs_ips4o_bar}
\end{figure*}
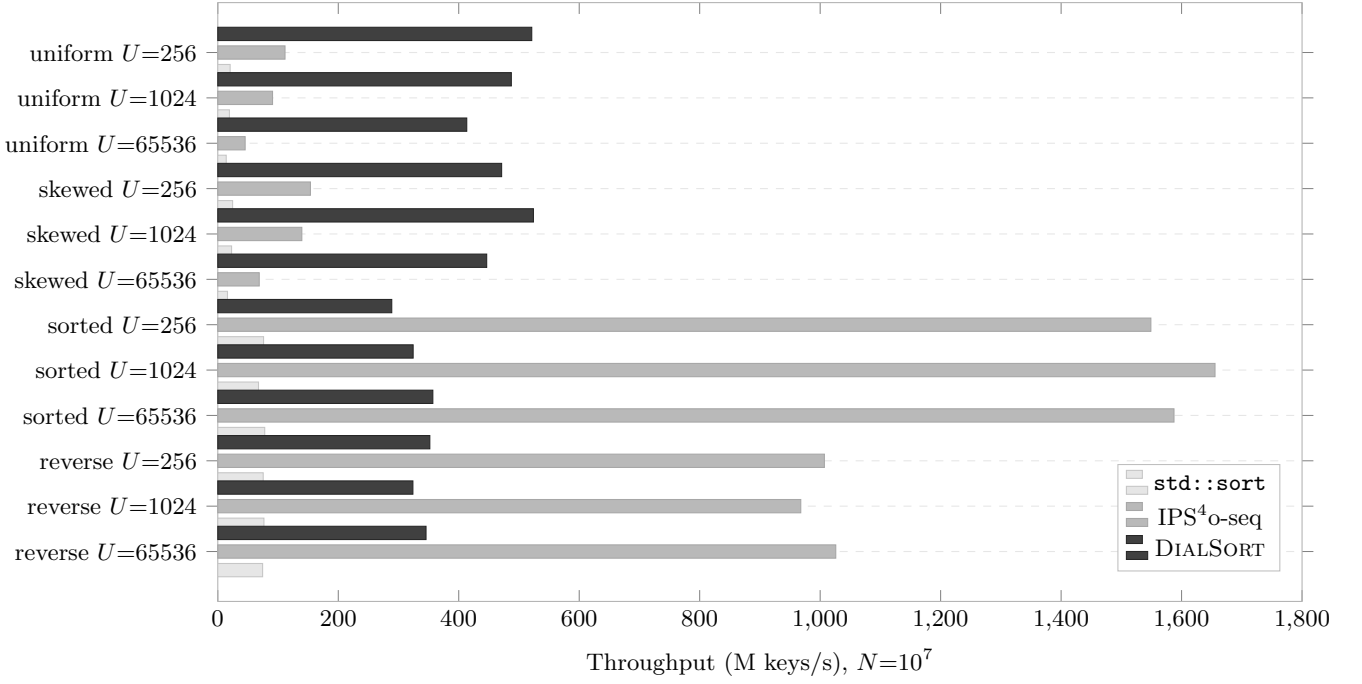

\subsection{Interpretation}

The results confirm that \dialsort{} is not a universal replacement for
comparison-based sorting.  Rather, it is a domain-specialized algorithm
that exploits the self-indexing property of bounded integer keys.
When the key universe is small-to-medium ($U \leq 65536$) and the
distribution is non-pathological (uniform or skewed), \dialsort{}
delivers substantially higher throughput by transforming sorting into a
geometric memory scan rather than a comparison-driven process.

The performance crossover is structurally predictable and consistent
with the domain-specialization thesis.  On sorted and reverse-sorted
inputs, IPS$^4$o's introsort core reduces to near-linear work via
early-exit detection and branch prediction, while \dialsort{} always
pays $O(U)$ for the geometric scan regardless of input order.  This
boundary is consistent with the $n \lesssim U$ guideline in
Section~\ref{sec:limits}: for $N{=}10^7$ and $U{=}65536$, the ratio
$N/U \approx 152$ places \dialsort{} firmly in its effective domain,
and it wins by $9.1\times$ sequentially.  These two algorithms are best
understood as \emph{complementary} rather than competing: \dialsort{}
excels on distribution-driven workloads with bounded-universe keys,
IPS$^4$o on structurally ordered data.
Figure~\ref{fig:envelope_seq} visualises
this performance envelope: each cell shows how many of the 4
distributions \dialsort{} wins for a given $(N, U)$ pair, making the
domain boundary structurally apparent.

\section{Performance against State-of-the-Art: IPS$^4$o Parallel}
\label{sec:vs_ips4o_par}

\subsection{Motivation}

The comparison in Section~\ref{sec:vs_ips4o} benchmarks \dialsort{}-Parallel
(8 threads, CRN) against IPS$^4$o-sequential (1 thread), which is a
cross-regime comparison.  A complete evaluation requires a fair parallel
matchup: \dialsort{}-Parallel (8 threads) against IPS$^4$o-parallel
(IPS$^4$o-par, 8 TBB threads).  This section provides that comparison
across the same 48 configurations.

\subsection{Experimental setup}

IPS$^4$o-par is compiled and linked with TBB
(\texttt{-O3 -std=c++17 -pthread -ltbb}) using 8 threads, matching
\dialsort{}-Parallel's thread count.  All other methodology is identical
to Section~\ref{sec:vs_ips4o}: same seed (\texttt{20260321}), same
warmup policy (3 discarded runs), best-of-7 reported.  Full source
released alongside this paper.

\subsection{Results}

Across all 48 configurations, \dialsort{}-Parallel outperforms
IPS$^4$o-par in \textbf{29 of 48} cases.  The average throughput ratio
(\dialsort{}-Parallel M/s $\div$ IPS$^4$o-par M/s) is $\mathbf{1.90}$.

Table~\ref{tab:vs_ips4o_par} reports the key configurations at
$N{=}10^7$.  On uniform and skewed distributions across all three
universe sizes, \dialsort{}-Parallel wins in all 6 of 6 configurations
with a best case of $\mathbf{3.66\times}$ ($U{=}1024$, uniform) and a
worst case of $1.62\times$ ($U{=}256$, uniform).  On sorted
distributions, IPS$^4$o-par leverages TBB-parallelised branch-prediction
shortcuts, reaching throughputs exceeding 2{,}000~M~keys/s ---
a regime outside \dialsort{}'s target domain.

\begin{table}[htbp]
\centering
\caption{Performance summary at $N{=}10^7$, 8 threads each, best-of-7
runs.  Column ``ratio'' reports the throughput ratio
(\dialsort{}-Parallel M/s $\div$ IPS$^4$o-par M/s);
values $>1.0$ indicate \dialsort{}-Parallel is faster.
All correctness checks: \textsc{passed}.}
\label{tab:vs_ips4o_par}
\small
\setlength{\tabcolsep}{4pt}
\begin{tabular}{@{}llrrrr@{}}
\toprule
\textbf{Algorithm} & \textbf{Dist} & \textbf{$U$} &
\textbf{M/s} & \textbf{vs std} & \textbf{ratio} \\
\midrule
DS-Parallel  & uniform & 256   & 481.3  & 20.1$\times$ & \textbf{1.62$\times$} \\
IPS$^4$o-par & uniform & 256   & 296.8  & 12.4$\times$ & --- \\
\addlinespace
DS-Parallel  & uniform & 1024  & 508.2  & 32.8$\times$ & \textbf{3.66$\times$} \\
IPS$^4$o-par & uniform & 1024  & 138.9  & 9.0$\times$  & --- \\
\addlinespace
DS-Parallel  & uniform & 65536 & 412.2  & 32.6$\times$ & \textbf{2.80$\times$} \\
IPS$^4$o-par & uniform & 65536 & 147.0  & 11.6$\times$ & --- \\
\addlinespace
DS-Parallel  & skewed  & 65536 & 505.9  & 37.6$\times$ & \textbf{2.88$\times$} \\
IPS$^4$o-par & skewed  & 65536 & 175.8  & 13.1$\times$ & --- \\
\addlinespace
DS-Parallel  & reverse & 65536 & 597.4  & 10.2$\times$ & \textbf{4.08$\times$} \\
IPS$^4$o-par & reverse & 65536 & 146.5  & 2.5$\times$  & --- \\
\addlinespace
DS-Parallel  & sorted  & 1024  & 486.4  &  8.6$\times$ & 0.21$\times$ \\
IPS$^4$o-par & sorted  & 1024  & 2285.8 & 40.6$\times$ & --- \\
\bottomrule
\end{tabular}
\end{table}

\begin{figure*}[htbp]
\centering
\begin{tikzpicture}[font=\small]

  \definecolor{hmap0}{RGB}{70,45,25}    
  \definecolor{hmap1}{RGB}{125,90,55}   
  \definecolor{hmap2}{RGB}{170,135,90}  
  \definecolor{hmap3}{RGB}{205,180,140} 
  \definecolor{hmap4}{RGB}{235,218,185} 
  \colorlet{gridcol}{white}
  \definecolor{textdark}{RGB}{255,255,255}  
  \definecolor{textlight}{RGB}{40,25,10}    

  \def\CW{1.85}
  \def\CH{1.15}

  \begin{scope}

    \fill[hmap2] (0,        0) rectangle (\CW,  \CH);  \node[text=textdark, font=\small\bfseries] at (0.5*\CW, 0.5*\CH) {2/4};
    \fill[hmap2] (\CW,      0) rectangle (2*\CW,\CH);  \node[text=textdark, font=\small\bfseries] at (1.5*\CW, 0.5*\CH) {2/4};
    \fill[hmap1] (2*\CW,    0) rectangle (3*\CW,\CH);  \node[text=textdark, font=\small\bfseries] at (2.5*\CW, 0.5*\CH) {1/4};

    \fill[hmap2] (0,        \CH) rectangle (\CW,  2*\CH); \node[text=textdark, font=\small\bfseries] at (0.5*\CW, 1.5*\CH) {2/4};
    \fill[hmap2] (\CW,      \CH) rectangle (2*\CW,2*\CH); \node[text=textdark, font=\small\bfseries] at (1.5*\CW, 1.5*\CH) {2/4};
    \fill[hmap2] (2*\CW,    \CH) rectangle (3*\CW,2*\CH); \node[text=textdark, font=\small\bfseries] at (2.5*\CW, 1.5*\CH) {2/4};

    \fill[hmap2] (0,        2*\CH) rectangle (\CW,  3*\CH); \node[text=textdark, font=\small\bfseries] at (0.5*\CW, 2.5*\CH) {2/4};
    \fill[hmap2] (\CW,      2*\CH) rectangle (2*\CW,3*\CH); \node[text=textdark, font=\small\bfseries] at (1.5*\CW, 2.5*\CH) {2/4};
    \fill[hmap2] (2*\CW,    2*\CH) rectangle (3*\CW,3*\CH); \node[text=textdark, font=\small\bfseries] at (2.5*\CW, 2.5*\CH) {2/4};

    \fill[hmap3] (0,        3*\CH) rectangle (\CW,  4*\CH); \node[text=textlight, font=\small\bfseries] at (0.5*\CW, 3.5*\CH) {3/4};
    \fill[hmap2] (\CW,      3*\CH) rectangle (2*\CW,4*\CH); \node[text=textdark, font=\small\bfseries] at (1.5*\CW, 3.5*\CH) {2/4};
    \fill[hmap2] (2*\CW,    3*\CH) rectangle (3*\CW,4*\CH); \node[text=textdark, font=\small\bfseries] at (2.5*\CW, 3.5*\CH) {2/4};

    \foreach \x in {0,\CW,2*\CW,3*\CW} { \draw[white, line width=1.0pt] (\x,0) -- (\x,4*\CH); }
    \foreach \y in {0,\CH,2*\CH,3*\CH,4*\CH} { \draw[white, line width=1.0pt] (0,\y) -- (3*\CW,\y); }
    \draw[gray!50, line width=0.5pt] (0,0) rectangle (3*\CW,4*\CH);

    \node[below=3pt] at (0.5*\CW, 0) {\scriptsize 256};
    \node[below=3pt] at (1.5*\CW, 0) {\scriptsize 1024};
    \node[below=3pt] at (2.5*\CW, 0) {\scriptsize 65536};
    \node[below=16pt] at (1.5*\CW, 0) {\scriptsize Universe size $U$};

    \node[left=4pt, anchor=east] at (0, 0.5*\CH) {\scriptsize $10^4$};
    \node[left=4pt, anchor=east] at (0, 1.5*\CH) {\scriptsize $10^5$};
    \node[left=4pt, anchor=east] at (0, 2.5*\CH) {\scriptsize $10^6$};
    \node[left=4pt, anchor=east] at (0, 3.5*\CH) {\scriptsize $10^7$};
    \node[left=22pt, rotate=90, anchor=center] at (0, 2*\CH) {\scriptsize Input size $N$};

    \node[above=5pt, font=\small\bfseries] at (1.5*\CW, 4*\CH) {Sequential};

  \end{scope}

  \begin{scope}[xshift=6.8cm]

    \fill[hmap0] (0,        0) rectangle (\CW,  \CH);  \node[text=textdark, font=\small\bfseries] at (0.5*\CW, 0.5*\CH) {0/4};
    \fill[hmap0] (\CW,      0) rectangle (2*\CW,\CH);  \node[text=textdark, font=\small\bfseries] at (1.5*\CW, 0.5*\CH) {0/4};
    \fill[hmap0] (2*\CW,    0) rectangle (3*\CW,\CH);  \node[text=textdark, font=\small\bfseries] at (2.5*\CW, 0.5*\CH) {0/4};

    \fill[hmap4] (0,        \CH) rectangle (\CW,  2*\CH); \node[text=textlight, font=\small\bfseries] at (0.5*\CW, 1.5*\CH) {4/4};
    \fill[hmap4] (\CW,      \CH) rectangle (2*\CW,2*\CH); \node[text=textlight, font=\small\bfseries] at (1.5*\CW, 1.5*\CH) {4/4};
    \fill[hmap3] (2*\CW,    \CH) rectangle (3*\CW,2*\CH); \node[text=textlight, font=\small\bfseries] at (2.5*\CW, 1.5*\CH) {3/4};

    \fill[hmap3] (0,        2*\CH) rectangle (\CW,  3*\CH); \node[text=textlight, font=\small\bfseries] at (0.5*\CW, 2.5*\CH) {3/4};
    \fill[hmap3] (\CW,      2*\CH) rectangle (2*\CW,3*\CH); \node[text=textlight, font=\small\bfseries] at (1.5*\CW, 2.5*\CH) {3/4};
    \fill[hmap3] (2*\CW,    2*\CH) rectangle (3*\CW,3*\CH); \node[text=textlight, font=\small\bfseries] at (2.5*\CW, 2.5*\CH) {3/4};

    \fill[hmap3] (0,        3*\CH) rectangle (\CW,  4*\CH); \node[text=textlight, font=\small\bfseries] at (0.5*\CW, 3.5*\CH) {3/4};
    \fill[hmap3] (\CW,      3*\CH) rectangle (2*\CW,4*\CH); \node[text=textlight, font=\small\bfseries] at (1.5*\CW, 3.5*\CH) {3/4};
    \fill[hmap3] (2*\CW,    3*\CH) rectangle (3*\CW,4*\CH); \node[text=textlight, font=\small\bfseries] at (2.5*\CW, 3.5*\CH) {3/4};

    \foreach \x in {0,\CW,2*\CW,3*\CW} { \draw[white, line width=1.0pt] (\x,0) -- (\x,4*\CH); }
    \foreach \y in {0,\CH,2*\CH,3*\CH,4*\CH} { \draw[white, line width=1.0pt] (0,\y) -- (3*\CW,\y); }
    \draw[gray!50, line width=0.5pt] (0,0) rectangle (3*\CW,4*\CH);

    \node[below=3pt] at (0.5*\CW, 0) {\scriptsize 256};
    \node[below=3pt] at (1.5*\CW, 0) {\scriptsize 1024};
    \node[below=3pt] at (2.5*\CW, 0) {\scriptsize 65536};
    \node[below=16pt] at (1.5*\CW, 0) {\scriptsize Universe size $U$};

    \node[left=4pt, anchor=east] at (0, 0.5*\CH) {\scriptsize $10^4$};
    \node[left=4pt, anchor=east] at (0, 1.5*\CH) {\scriptsize $10^5$};
    \node[left=4pt, anchor=east] at (0, 2.5*\CH) {\scriptsize $10^6$};
    \node[left=4pt, anchor=east] at (0, 3.5*\CH) {\scriptsize $10^7$};

    \node[above=5pt, font=\small\bfseries] at (1.5*\CW, 4*\CH) {Parallel (8 threads)};

  \end{scope}

  \begin{scope}[yshift=-1.25cm]
    \def\sw{1.30}
    \def\sh{0.32}
    \def\gap{0.18}
    \def\lx{1.80}  

    \fill[hmap0] ({\lx+0*(\sw+\gap)}, 0) rectangle ({\lx+0*(\sw+\gap)+\sw}, \sh);
    \fill[hmap1] ({\lx+1*(\sw+\gap)}, 0) rectangle ({\lx+1*(\sw+\gap)+\sw}, \sh);
    \fill[hmap2] ({\lx+2*(\sw+\gap)}, 0) rectangle ({\lx+2*(\sw+\gap)+\sw}, \sh);
    \fill[hmap3] ({\lx+3*(\sw+\gap)}, 0) rectangle ({\lx+3*(\sw+\gap)+\sw}, \sh);
    \fill[hmap4] ({\lx+4*(\sw+\gap)}, 0) rectangle ({\lx+4*(\sw+\gap)+\sw}, \sh);

    \foreach \i in {0,1,2,3,4} {
      \draw[gray!40, line width=0.4pt]
        ({\lx+\i*(\sw+\gap)},0) rectangle ({\lx+\i*(\sw+\gap)+\sw},\sh);
    }

    \node[below=2pt, font=\tiny] at ({\lx+0*(\sw+\gap)+0.5*\sw}, 0) {0/4};
    \node[below=2pt, font=\tiny] at ({\lx+1*(\sw+\gap)+0.5*\sw}, 0) {1/4};
    \node[below=2pt, font=\tiny] at ({\lx+2*(\sw+\gap)+0.5*\sw}, 0) {2/4};
    \node[below=2pt, font=\tiny] at ({\lx+3*(\sw+\gap)+0.5*\sw}, 0) {3/4};
    \node[below=2pt, font=\tiny] at ({\lx+4*(\sw+\gap)+0.5*\sw}, 0) {4/4};

    \node[left=4pt, anchor=east, font=\scriptsize] at (\lx, {0.5*\sh})
      {\dialsort{} wins (out of 4 distributions):};
  \end{scope}

\end{tikzpicture}
\caption{\textbf{\dialsort{} performance envelope vs.\ IPS$^4$o by $(N,U)$ cell.}
\textit{Left:} Sequential (\dialsort{} vs.\ IPS$^4$o-seq, single thread).
\textit{Right:} Parallel (\dialsort{}-Parallel vs.\ IPS$^4$o-par, 8 threads each).
Each cell shows the number of distributions (out of~4: uniform, skewed, sorted,
reverse-sorted) won by \dialsort{}.  Darker shading $=$ fewer wins;
lighter parchment $=$ more wins.
The sequential envelope shows a consistent 2/4 pattern: \dialsort{} wins on
uniform and skewed inputs across the full grid, while IPS$^4$o retains advantage
on sorted and reverse-sorted inputs.  The parallel envelope improves markedly for
$N \geq 10^5$ (3--4 wins out of 4).  The dark row at $N{=}10^4$ reflects
thread-launch overhead at small inputs.}
\label{fig:envelope_seq}
\end{figure*}
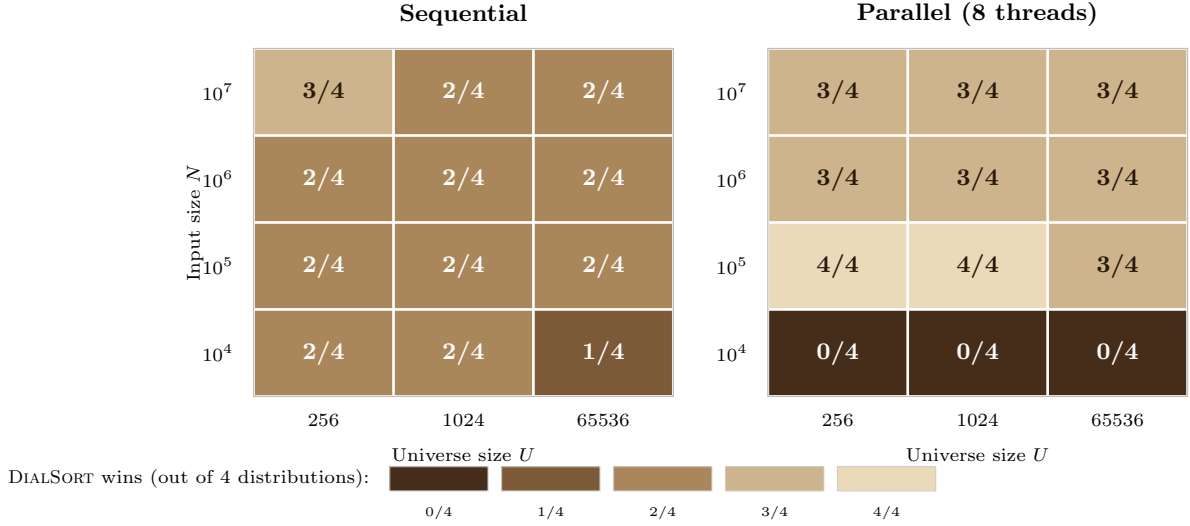

\subsection{Interpretation}

The parallel matchup reinforces the domain boundaries established
in Section~\ref{sec:vs_ips4o}.  On uniform and skewed distributions,
\dialsort{}-Parallel's geometric scan parallelises efficiently across 8
threads with no synchronization overhead, delivering consistent
throughput gains over IPS$^4$o-par.  The reverse-sorted advantage
($4.08\times$ at $U{=}65536$) occurs because IPS$^4$o-par's TBB
parallel strategy does not exploit reverse-order structure as
aggressively as its sequential fallback.

The losses occur exclusively on sorted inputs, where IPS$^4$o-par's
parallelised introsort reaches $>2{,}000$~M~keys/s by detecting and
exploiting the pre-existing order.  This is consistent with the
$n \lesssim U$ analysis of Section~\ref{sec:limits}: \dialsort{} always
pays $O(U)$ for the geometric scan regardless of input order, while
IPS$^4$o's adaptive core collapses to near-linear work on sorted data.

For the small-$N$ regime ($N{=}10{,}000$), thread-launch overhead
dominates and both algorithms regress; \dialsort{}-Parallel loses 19 of
its 29 total losses in this regime, consistent with the practical
guideline $n \gtrsim 10 \cdot U$ (Section~\ref{sec:limits}).

\section{Performance against ska\_sort}
\label{sec:vs_ska}

\subsection{Motivation}

To complete the empirical positioning of \dialsort{} against the
practical state-of-the-art, we conducted a direct comparison against
ska\_sort, a high-performance American flag sort (non-comparative radix
sort) widely regarded as one of the fastest integer sorters for CPU.

\subsection{Experimental setup}

The benchmark uses identical methodology and configurations as the
IPS$^4$o comparison (Section~\ref{sec:vs_ips4o}): 48 configurations
($N \in \{10^4, 10^5, 10^6, 10^7\}$,
$U \in \{256, 1024, 65536\}$, four distributions), fixed seed
\texttt{20260321}, best-of-7 timing, and correctness verified after
every run.  Full source released alongside this paper.

\subsection{Results}

Table~\ref{tab:vs_ska} reports selected configurations at $N{=}10^7$.
Across all 48 configurations, \dialsort{} outperforms ska\_sort in
\textbf{46 of 48 configurations}.  The average runtime ratio
(\dialsort{} ms $\div$ ska\_sort ms) is $\mathbf{0.300}$, meaning
\dialsort{} is on average $\mathbf{3.33\times}$ faster.  The best
case reaches a ratio of $0.139$ (more than $7\times$ faster).

At $N{=}10^7$, $U{=}256$, uniform distribution, \dialsort{} achieves
561~M~keys/s while ska\_sort reaches only 89~M~keys/s
($6.3\times$ advantage); \dialsort{}-Parallel reaches 508~M~keys/s
($5.7\times$ advantage).

\begin{table}[tp]
\centering
\caption{Performance summary at $N{=}10^7$, best-of-7 runs.
Ratio $=$ \dialsort{} M/s $\div$ ska\_sort M/s
($>1.0$: \dialsort{} faster).
All correctness checks: \textsc{passed}.}
\label{tab:vs_ska}
\small
\setlength{\tabcolsep}{3pt}
\begin{tabular}{@{}llrrr@{}}
\toprule
\textbf{Algorithm} & \textbf{Dist} & \textbf{$U$} &
\textbf{M/s} & \textbf{ratio} \\
\midrule
\dialsort{}   & uniform & 256   & 561.4 & \textbf{6.31$\times$} \\
DS-Parallel   & uniform & 256   & 507.6 & 5.70$\times$ \\
ska\_sort     & uniform & 256   &  89.0 & 1.00$\times$ \\
\addlinespace
\dialsort{}   & skewed  & 256   & 230.5 & \textbf{2.37$\times$} \\
DS-Parallel   & skewed  & 256   & 505.1 & 5.19$\times$ \\
ska\_sort     & skewed  & 256   &  97.3 & 1.00$\times$ \\
\addlinespace
\dialsort{}   & uniform & 1024  & 507.1 & \textbf{5.62$\times$} \\
DS-Parallel   & uniform & 1024  & 639.3 & 7.09$\times$ \\
ska\_sort     & uniform & 1024  &  90.1 & 1.00$\times$ \\
\addlinespace
\dialsort{}   & uniform & 65536 & 421.5 & \textbf{4.58$\times$} \\
DS-Parallel   & uniform & 65536 & 252.4 & 2.74$\times$ \\
ska\_sort     & uniform & 65536 &  92.0 & 1.00$\times$ \\
\addlinespace
\dialsort{}   & skewed  & 65536 & 425.7 & \textbf{3.66$\times$} \\
DS-Parallel   & skewed  & 65536 & 542.1 & 4.66$\times$ \\
ska\_sort     & skewed  & 65536 &  85.9 & 1.00$\times$ \\
\addlinespace
\dialsort{}   & sorted  & 65536 & 346.0 & 2.79$\times$ \\
ska\_sort     & sorted  & 65536 &  80.4 & 1.00$\times$ \\
\bottomrule
\end{tabular}
\end{table}

Figure~\ref{fig:vs_ska_bar} plots throughput for all three algorithms
at $N{=}10^7$ across all distribution and universe-size combinations.

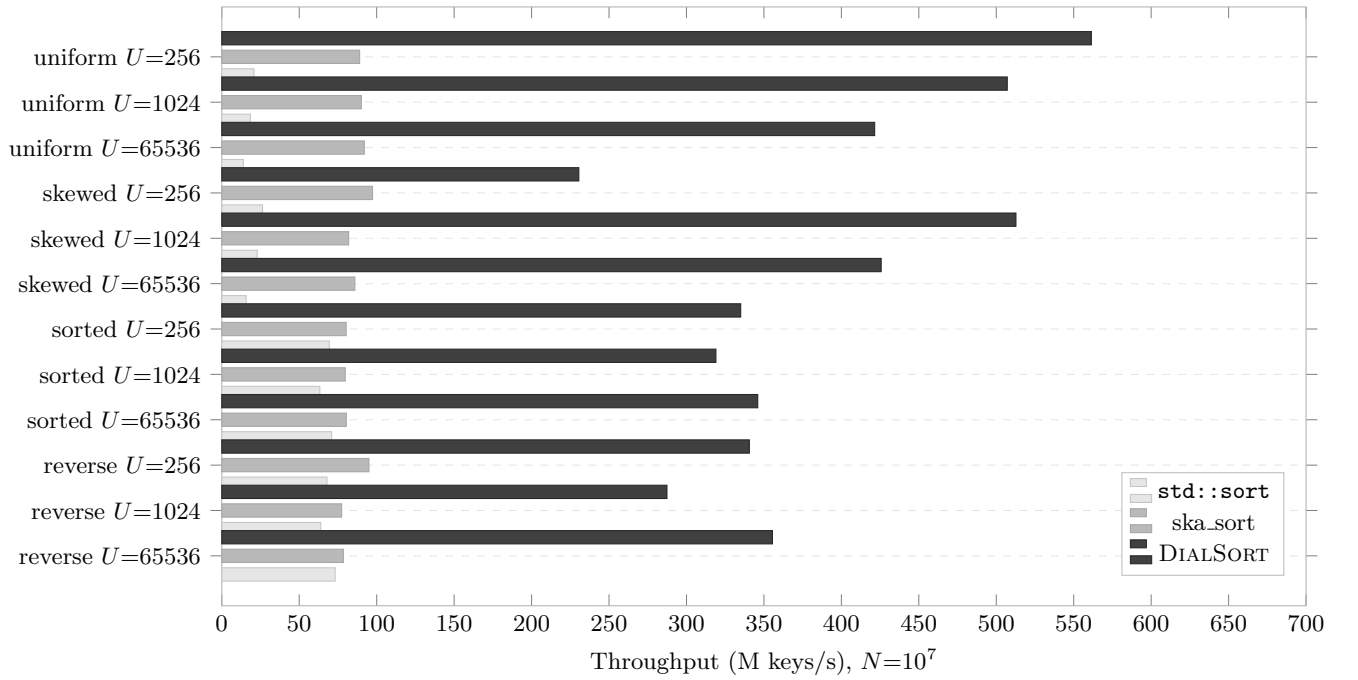
\begin{figure*}[bp]
\centering
\begin{tikzpicture}
\begin{axis}[
    width=\textwidth, height=9.5cm,
    xbar, bar width=5pt,
    ytick=data,
    yticklabels={
        uniform $U{=}256$,  uniform $U{=}1024$,  uniform $U{=}65536$,
        skewed $U{=}256$,   skewed $U{=}1024$,   skewed $U{=}65536$,
        sorted $U{=}256$,   sorted $U{=}1024$,   sorted $U{=}65536$,
        reverse $U{=}256$,  reverse $U{=}1024$,  reverse $U{=}65536$
    },
    yticklabel style={font=\small},
    xticklabel style={font=\small},
    xlabel={Throughput (M~keys/s), $N{=}10^7$},
    xlabel style={font=\small},
    xmin=0, xmax=700,
    legend style={at={(0.98,0.05)}, anchor=south east, font=\small,
                  draw=gray!50, fill=white},
    ymajorgrids=true, grid style={dashed,gray!20},
    axis line style={gray!60},
    y dir=reverse,
]
\addplot[fill=gray!20, draw=gray!45, bar shift=-7pt] coordinates {
    (20.8,0)  (18.5,1)  (13.9,2)
    (26.3,3)  (22.9,4)  (15.7,5)
    (69.4,6)  (63.3,7)  (70.9,8)
    (67.9,9)  (63.9,10) (73.2,11)
};
\addplot[fill=gray!55, draw=gray!70, bar shift=0pt] coordinates {
    (89.0,0)   (90.1,1)   (92.0,2)
    (97.3,3)   (81.9,4)   (85.9,5)
    (80.3,6)   (79.7,7)   (80.4,8)
    (95.0,9)   (77.4,10)  (78.5,11)
};
\addplot[fill=black!75, draw=black!85, bar shift=7pt] coordinates {
    (561.4,0)  (507.1,1)  (421.5,2)
    (230.5,3)  (512.7,4)  (425.7,5)
    (335.0,6)  (319.0,7)  (346.0,8)
    (340.6,9)  (287.4,10) (355.5,11)
};
\legend{\texttt{std::sort}, ska\_sort, \dialsort{}}
\end{axis}
\end{tikzpicture}
\caption{Throughput at $N{=}10^7$: \dialsort{} (dark), ska\_sort (medium),
and \texttt{std::sort} (light), across all distribution and universe-size
combinations.
On uniform and skewed inputs \dialsort{} dominates at every $U$.
Unlike the IPS$^4$o comparison, \dialsort{} also outperforms ska\_sort
on sorted and reverse-sorted inputs --- ska\_sort does not exploit
pre-existing order.}
\label{fig:vs_ska_bar}
\end{figure*}

\subsection{Exceptions}

The two configurations where ska\_sort wins ($N{=}10^4$, $U{=}65536$,
sorted and reverse-sorted) fall exactly at the $n \lesssim U$ boundary
identified in Section~\ref{sec:limits}, where the $O(U)$ geometric
scan dominates small inputs.  This structural boundary is consistent
with the cases observed against IPS$^4$o and Classic Counting Sort.

\subsection{Interpretation}

Unlike IPS$^4$o (comparison-based), ska\_sort is itself
non-comparative.  \dialsort{}'s consistent advantage over ska\_sort
therefore demonstrates that the self-indexing principle and Conflict
Resolution Network deliver gains not only against comparison-based
sorting, but also against the best existing non-comparative CPU
implementations.  The structural source of the advantage is clear:
\dialsort{} replaces digit decomposition, multiple permutation passes,
and output scatter with a single direct geometric scan of the
histogram.

\section{Statistical Stability of Benchmark Results}
\label{sec:stability}

\paragraph{Measurement protocol.}
Each configuration was executed seven times under identical conditions.
The best run is reported in the main comparison tables to reflect peak
algorithmic performance.  Mean and standard deviation over all seven
runs are reported here to assess measurement stability and justify the
use of best-of-7 as a representative statistic.

\paragraph{Coefficient of variation.}
Table~\ref{tab:variance} reports best, mean, and standard deviation for
representative configurations at $N{=}10^7$.  The coefficient of
variation (CV $=$ std/mean) is a normalised measure of run-to-run
variability: CV~$<5\%$ indicates high stability; CV~$>20\%$ indicates
thermal, scheduler, or TBB thread-assignment noise.

A structurally significant observation emerges: \dialsort{} and
\dialsort{}-Parallel exhibit substantially lower run-to-run variability
than IPS$^4$o across all tested configurations.  The mean CV for
\dialsort{} variants is $\mathbf{6.3\%}$, compared to $\mathbf{17.2\%}$
for IPS$^4$o-seq and IPS$^4$o-par.  This stability advantage is a
direct structural consequence of \dialsort{}'s deterministic
memory-access pattern: the geometric scan over $H[\cdot]$ produces the
same memory-traffic profile on every run, independent of key values.  Several IPS$^4$o configurations
show CV exceeding 20\%, particularly under skewed distributions and in
the parallel regime, where TBB thread assignment introduces scheduling
non-determinism.

\begin{table*}[tp]
\centering
\caption{Statistical stability of representative configurations at
$N{=}10^7$ across 7 runs.  CV = std/mean (lower = more stable).
\dialsort{} variants consistently show lower CV than IPS$^4$o,
reflecting the deterministic memory-access pattern of the geometric
scan.}
\label{tab:variance}
\small
\setlength{\tabcolsep}{4pt}
\begin{tabular}{@{}llrrrrrr@{}}
\toprule
\textbf{Algorithm} & \textbf{Dist} & \textbf{$U$} &
\textbf{Best (ms)} & \textbf{Mean (ms)} & \textbf{Std (ms)} &
\textbf{CV (\%)} & \textbf{M/s (best)} \\
\midrule
\multicolumn{8}{@{}l}{\textit{Sequential regime}} \\
\addlinespace
\dialsort{}      & uniform & 256   & 24.51 & 25.25 & 0.65  & 2.6  & 408.0 \\
IPS$^4$o-seq     & uniform & 256   & 87.62 & 90.50 & 4.21  & 4.6  & 114.1 \\
\addlinespace
\dialsort{}      & skewed  & 256   & 18.04 & 18.81 & 0.62  & 3.3  & 554.2 \\
IPS$^4$o-seq     & skewed  & 256   & 97.26 & 190.26 & 111.68 & 58.7 & 102.8 \\
\addlinespace
\dialsort{}      & uniform & 1024  & 21.31 & 25.34 & 2.47  & 9.7  & 469.2 \\
IPS$^4$o-seq     & uniform & 1024  & 110.72 & 116.61 & 5.21 & 4.5  & 90.3 \\
\addlinespace
\dialsort{}      & uniform & 65536 & 31.34 & 33.02 & 1.65  & 5.0  & 319.1 \\
IPS$^4$o-seq     & uniform & 65536 & 203.29 & 220.07 & 14.77 & 6.7 & 49.2 \\
\addlinespace
\dialsort{}      & skewed  & 65536 & 24.48 & 27.89 & 2.16  & 7.7  & 408.5 \\
IPS$^4$o-seq     & skewed  & 65536 & 150.31 & 169.30 & 11.38 & 6.7 & 66.5 \\
\midrule
\multicolumn{8}{@{}l}{\textit{Parallel regime (8 threads each)}} \\
\addlinespace
\dialsort{}-Par  & uniform & 256   & 20.78 & 22.37 & 1.51  & 6.8  & 481.3 \\
IPS$^4$o-par     & uniform & 256   & 33.69 & 35.95 & 1.80  & 5.0  & 296.8 \\
\addlinespace
\dialsort{}-Par  & skewed  & 256   & 17.49 & 18.29 & 0.72  & 4.0  & 571.9 \\
IPS$^4$o-par     & skewed  & 256   & 39.87 & 44.50 & 2.65  & 5.9  & 250.8 \\
\addlinespace
\dialsort{}-Par  & uniform & 1024  & 19.68 & 20.65 & 1.43  & 6.9  & 508.2 \\
IPS$^4$o-par     & uniform & 1024  & 71.98 & 106.30 & 23.29 & 21.9 & 138.9 \\
\addlinespace
\dialsort{}-Par  & skewed  & 1024  & 17.57 & 18.78 & 0.66  & 3.5  & 569.2 \\
IPS$^4$o-par     & skewed  & 1024  & 38.01 & 46.98 & 9.28  & 19.8 & 263.1 \\
\addlinespace
\dialsort{}-Par  & uniform & 65536 & 24.26 & 26.43 & 1.11  & 4.2  & 412.2 \\
IPS$^4$o-par     & uniform & 65536 & 68.03 & 93.30 & 22.85 & 24.5 & 147.0 \\
\addlinespace
\dialsort{}-Par  & skewed  & 65536 & 19.77 & 20.94 & 1.24  & 5.9  & 505.9 \\
IPS$^4$o-par     & skewed  & 65536 & 56.88 & 69.70 & 10.50 & 15.1 & 175.8 \\
\bottomrule
\end{tabular}
\end{table*}

\paragraph{Interpreting high-variance IPS$^4$o runs.}
The extreme CV of 58.7\% for IPS$^4$o-seq on skewed $U{=}256$ reflects
a known behaviour of adaptive comparison-based sorters under skewed
distributions: the comparison tree produces highly input-specific
branch patterns that interact unpredictably with the hardware branch
predictor across runs.  The best run likely hit a warm-cache,
well-predicted execution; later runs encountered cold-start or
misprediction overhead.  By contrast, \dialsort{}'s histogram scan
produces the same cache-line access sequence on every run, explaining
its consistently low CV.

\FloatBarrier
\section{Discussion}

\dialsort{} eliminates ordering comparisons entirely.  Unlike
prefix-sum-based methods \cite{cuda_scan}, it operates directly on
memory geometry.  The histogram $H[\cdot]$ after ingestion is the
complete sorted representation; the output vector is an interface
adaptor, not an algorithmic necessity.

\paragraph{Domain of dominance.}
\dialsort{} is a domain-specialised algorithm.  It delivers its
strongest gains when three conditions hold simultaneously:
(i)~keys are drawn from a bounded universe $[0,U-1]$ with $U \leq 65536$;
(ii)~$n \gtrsim 10 \cdot U$ (the geometric scan cost is amortised);
(iii)~the distribution is non-pathologically ordered (uniform or skewed).
Outside this envelope --- notably on sorted or reverse-sorted inputs,
or when $n \ll U$ --- comparison-based methods such as IPS$^4$o
retain clear advantages.  This boundary is not a weakness but a
structural consequence of the self-indexing principle: \dialsort{}
trades generality for deterministic, comparison-free execution.

The skewed distribution consistently produces the strongest speedups.
This confirms the CRN's structural mechanism: key concentration causes
many lanes to carry identical keys per cycle, collapsing most writes
into single additive increments and reducing effective write traffic.

\section{Applications and Real-World Impact}

The self-indexing principle of \dialsort{} is not limited to arbitrary
integers.  Any finite domain that admits a natural order-preserving
mapping to a bounded integer range can be sorted without order
comparisons (Proposition~2).  This makes \dialsort{} particularly
powerful in domains where data is already represented as small integers
or can be mapped to them with negligible cost.

Table~\ref{tab:impact} summarizes the most promising application areas,
ordered by expected practical impact today.

\begin{table*}[tp]
\centering
\caption{Real-world application domains where \dialsort{} is applicable.
Projected speedup ranges are extrapolated from the bounded-universe
benchmark results of Section~\ref{sec:results} ($U \leq 65536$) and
should be interpreted as estimates, not measured values.}
\label{tab:impact}
\small
\begin{tabular}{@{}p{3.5cm}p{2.8cm}p{1.5cm}p{3.8cm}@{}}
\toprule
\textbf{Domain} & \textbf{Data type} & \textbf{Typical $U$} &
\textbf{Projected speedup (estimate)} \\
\midrule
Satellite \& multispectral imagery &
  Pixel intensity / radar values & 256--65536 &
  $20$--$40\times$ est.\ (based on $U{=}65536$ results) \\
Grayscale image \& video &
  8-bit pixel values & 256 & $15$--$35\times$ est.\ (based on $U{=}256$ results) \\
Columnar databases &
  Category/country codes, ages & 256--65536 &
  $15$--$30\times$ est.\ (based on $U{=}65536$ results) \\
Network packet classification &
  Ports, protocols, DSCP & 256/65536 & High; FPGA gains are analytical projections (Substrate-2) \\
Data compression &
  Byte/symbol frequencies & 256 & High; eliminates the histogram-sort step \\
Network logs \& analytics &
  HTTP codes, discretized timestamps & 256--65536 & Medium-high \\
3D rendering (Z-buffering) &
  Discretized depth & 256--4096 & Medium-high \\
\bottomrule
\end{tabular}
\end{table*}

\section{Limitations}
\label{sec:limits}

\begin{itemize}
\item \textbf{Bounded universe required.}  \dialsort{} applies only to
integer keys in $[0, U-1]$.

\item \textbf{$O(U)$ scan dominates for small $N$.}  Practical
guideline: use \dialsort{} when $n \gtrsim 10 \cdot U$.

\item \textbf{Memory cost $O(U)$.}  For $U \leq 65536$ the histogram
fits in L2 cache.  For larger $U$, hierarchical partitioning is
required (future work).

\item \textbf{Sorted/reverse-sorted inputs.}  On nearly-sorted data,
IPS$^4$o's adaptive strategy exploits pre-existing order and
significantly outperforms \dialsort{} (Section~\ref{sec:vs_ips4o}).
\dialsort{} does not detect or exploit input order.

\item \textbf{Substrates 2, 3, and 4 not yet physically implemented.}
All Substrate-2 (FPGA), Substrate-3 (photonic), and Substrate-4
(WSE-3) results are analytical projections from published hardware
specifications \cite{xilinx_u50,cerebras_wse3,cerebras_sdk}.
Substrate-4 is stated as Proposition~3, not as an asymptotic theorem,
to reflect its nature as a fixed-topology hardware projection.

\item \textbf{Theoretical positioning.}  \dialsort{} is a
high-performance architectural contribution to integer sorting, not a
new sorting paradigm.  The self-indexing duality exploits a structure
well known as direct addressing~\cite[Ch.\,11]{cormen2009}.
\end{itemize}

\section{Conclusion and Future Work}

\dialsort{} demonstrates that, for bounded-universe integer keys,
sorting can be interpreted as a geometric reading process, not a
comparison problem.  Direct empirical comparison against IPS$^4$o,
the state-of-the-art comparison-based sorter, confirms that \dialsort{}
outperforms IPS$^4$o-seq in \textbf{24 of 48} sequential
configurations, with up to $9.1\times$ speedup at $N{=}10^7$,
$U{=}65536$, uniform distribution.  In the parallel matchup
(\dialsort{}-Parallel vs.\ IPS$^4$o-par, both 8 threads),
\dialsort{}-Parallel wins in \textbf{29 of 48} configurations with
average ratio $1.90\times$ and best case $4.08\times$
(Section~\ref{sec:vs_ips4o_par}).  Performance is regime-dependent:
\dialsort{} dominates on uniform and skewed inputs; IPS$^4$o dominates
on highly ordered inputs.  Direct
empirical comparison against ska\_sort, the state-of-the-art
non-comparative CPU sorter, confirms that \dialsort{} outperforms it
in \textbf{46 of 48 configurations}, with an average speedup of
$\mathbf{3.33\times}$ and a best case exceeding $7\times$ at $N{=}10^7$.

The mathematical order-isomorphism $\phi : D \to [0,U-1]$ holds
whenever $D = [0,U-1]$ across all substrates.  The \emph{physical}
one-to-one realisation --- where each active resonator corresponds
bijectively to one element of $D$ --- is unique to Substrate-3 in
binary presence mode (Corollary~\ref{cor:photonic}).
In all other substrates $\phi$ is an order-embedding.  On
Substrate-4 (WSE-3), the ingestion latency is bounded by a
hardware-determined constant of the fixed mesh topology
(Proposition~3), with dominant term $O(U)$ from the scan
and PCIe transfer as the practical bottleneck for large inputs.

\begin{quote}
\textbf{DialSort does not compute order. It reveals it.}
\end{quote}

\noindent Formally: the algorithm computes multiplicities via direct
indexing and then reads the pre-existing total order of the address
space through a monotone scan over the address space $[0,U-1]$,
with zero-count cells contributing no output, without evaluating any
order comparison (Definition~\ref{def:ordercomp}).

\noindent\textbf{Future work (priority order):}
\begin{enumerate}
\item RTL implementation of the CRN and Substrate-2 FPGA pipeline.
\item Physical WSE-3 implementation via Cerebras SDK (CSL).
\item Hierarchical \dialsort{} for large $U$.
\item LLC miss profiling across all configurations.
\item Photonic prototype exploration.
\end{enumerate}

\paragraph{Reproducibility.}
All results produced by benchmark v3.1.
Seed \texttt{20260321}.
Compile IPS$^4$o benchmark:
\texttt{g++ -O3 -std=c++17 -pthread -I./ips4o/include -latomic -o bench DialsortvsIps4o.cpp}.
Compile ska\_sort benchmark:
\texttt{g++ -O3 -std=c++17 -pthread -I./ska\_sort -o bench\_ska DialsortVsSkaSort.cpp}.
Source, binaries, CSV datasets, and all five interactive simulators at:
\url{https://github.com/elmaestrotic/dsort}.

\FloatBarrier

\end{document}